\newif\ifcslaw

\ifcslaw
    \documentclass[acmsmall,nonacm]{acmart}
\else
    \documentclass[11pt]{article}
    \usepackage{amsmath,natbib}
    \usepackage[margin=0.85in]{geometry}
\fi

\usepackage{algorithm}
\usepackage{algorithmic}

\usepackage{amsthm,amssymb}
\usepackage{hyperref}
\usepackage{cleveref}
\usepackage{url}
\usepackage{graphicx}
\usepackage{thm-restate}
\usepackage{booktabs}
\usepackage{tabularx}

\newcommand{\mr}[1]{\textcolor{magenta}{[MR: #1]}}
\newcommand{\proposedchange}[1]{{#1}}
\newcommand{\ch}[1]{\textcolor{olive}{[CH: #1]}}
\newcommand{\bl}[1]{\textcolor{blue}{[BL: #1]}}
\newcommand{\sib}[1]{\textcolor{cyan}{[SIB: #1]}}

\renewcommand{\mr}[1]{}
\renewcommand{\ch}[1]{}
\renewcommand{\bl}[1]{}
\renewcommand{\sib}[1]{}

\DeclareMathOperator{\Ei}{Ei}
\newcommand{\dif}{\mathop{}\!\mathrm{d}}
\def\mc{\mathcal}

\renewcommand{\b}[1]{\left[#1\right]}
\newcommand{\E}{\mathbb{E}}
\newcommand{\Eb}[1]{\mathbb{E}\b{#1}}
\newcommand{\EE}[2]{\mathbb{E}_{#1}\left[#2\right]}
\newcommand{\iidsim}{\overset{\text{iid}}{\sim}}

\def\given{\; | \;}
\newcommand\numberthis{\addtocounter{equation}{1}\tag{\theequation}}
\newcommand{\ind}{\mathbb{I}}
\newcommand{\indic}[1]{\mathbb{I}\left\{{#1} \right\}}

\def\ceil#1{\lceil #1 \rceil}
\def\floor#1{\lfloor #1 \rfloor}
\newcommand{\paren}[1]{\left( {#1} \right)}

\newcommand{\supp}[1]{\texttt{supp}\left( {#1} \right)}
\newcommand{\sqparen}[1]{\left[ {#1} \right]}
\newcommand{\median}{\textnormal{\texttt{median}}}

\newcommand{\defeq}{\triangleq}
\newcommand{\abs}[1]{\left| #1 \right|}

\def\mubound{\bar{\mu}^{\text{bounded}}}
\def\mueb{\bar{\mu}^{\text{eb}}}

\def\mumono{\bar{\mu}^{\text{mono}}}

\theoremstyle{plain}
\newtheorem{theorem}{Theorem}[section]
\crefname{theorem}{Theorem}{Theorems}
\newtheorem{proposition}[theorem]{Proposition}
\newtheorem{lemma}[theorem]{Lemma}
\newtheorem{corollary}[theorem]{Corollary}
\theoremstyle{definition}
\newtheorem{definition}[theorem]{Definition}
\newtheorem{assumption}[theorem]{Assumption}
\theoremstyle{remark}

\crefname{assumption}{Assumption}{Assumptions}
\Crefname{assumption}{Assumption}{Assumptions}
\crefname{corollary}{Corollary}{Corollary}
\Crefname{corollary}{Corollary}{Corollary}
\newcounter{muassumption}
\renewcommand{\themuassumption}{A\arabic{muassumption}}
\newcommand{\assumplabel}[1]{\leavevmode\refstepcounter{muassumption}(\themuassumption)\label{#1}}

\crefname{muassumption}{assumption}{assumptions} 
\Crefname{muassumption}{Assumption}{Assumptions}

\newcommand{\N}{\mathbb{N}}

\newcommand{\cA}{\mathcal{A}}

\newcommand{\cD}{\mathcal{D}}
\newcommand{\cE}{\mathcal{E}}
\newcommand{\cF}{\mathcal{F}}

\newcommand{\cP}{\mathcal{P}}

\newcommand{\cX}{\mathcal{X}}

\newcommand{\cY}{\mathcal{Y}}

\newcommand{\bbP}{\mathbb{P}}

\newcommand{\Dtrain}{D^{\mathrm{train}}}
\newcommand{\Dtest}{D^{\mathrm{test}}}

\ifcslaw
    \acmConference[Conference acronym 'XX]{Make sure to enter the correct
      conference title from your rights confirmation email}{June 03--05,
      2018}{Woodstock, NY}
\else
    
\begin{document}
        \title{Statistical Guarantees in the Search for Less Discriminatory Algorithms}
        \author{Chris Hays\thanks{MIT}
        \and
        Ben Laufer\thanks{Cornell University}
        \and
        Solon Barocas\thanks{Microsoft Research}
        \and
        Manish Raghavan\thanks{MIT}}
    \maketitle
\fi

\begin{abstract}
\proposedchange{U.S. discrimination law can impose liability on firms that fail to adopt a less discriminatory alternative (LDA): a decision policy that achieves the same business objectives while reducing disparate impact on legally protected groups. Recent scholarship argues that this doctrine has direct implications for algorithmic decision-making in high-stakes domains such as employment, lending, and housing, potentially obligating firms to search for “less discriminatory algorithms” \citep{black_less_2024}. Regulators have at times encouraged proactive LDA searches, reinforcing the expectation of a good-faith effort to identify equally performant models with lower disparate impact.}

\proposedchange{Model multiplicity makes such searches plausible: retraining with different random seeds can yield models with comparable predictive performance but materially different disparate impacts. Yet firms cannot retrain indefinitely, raising a central question: when is the search sufficient to demonstrate good faith?}

We formalize LDA search under multiplicity as an optimal stopping problem in which a developer seeks to produce evidence that further search is unlikely to yield meaningful improvements. Our main contribution is an adaptive stopping algorithm that provides a high-probability upper bound on the best disparate-impact gains attainable through continued retraining, enabling developers to certify (e.g., to a court) that additional search is unlikely to help. We also show how stronger distributional assumptions over the model space can yield tighter bounds, and we validate the approach on real-world credit and housing datasets.

\end{abstract}

\ifcslaw
    \begin{document}
\fi

\ifcslaw
\title{Statistical Guarantees in the Search for Less Discriminatory Algorithms}
\fi

\ifcslaw
    \author{Chris Hays}
    \email{{jhays@mit.edu}}
    \affiliation{%
      \institution{MIT}
      \city{Cambridge}
      \state{MA}
      \country{USA}
    }
    \author{Ben Laufer}
    \affiliation{%
      \institution{Cornell Tech}
      \city{New York}
      \state{NY}
      \country{USA}
    }
    \author{Solon Barocas}
    \affiliation{%
      \institution{Microsoft Research}
      \city{New York}
      \state{NY}
      \country{USA}
    }
    \author{Manish Raghavan}
    \orcid{1234-5678-9012}
    \affiliation{%
      \institution{MIT}
      \city{Cambridge}
      \state{MA}
      \country{USA}
    }
\else
\fi




\ifcslaw
    \keywords{less discriminatory algorithms, sequential decision-making, anytime-valid inference}
\fi

\ifcslaw
\maketitle
\fi

\section{Introduction}



Data-driven models increasingly underpin decision making in critical domains like employment, credit, and housing. While these models have been embraced for their potential to improve the quality and efficiency of such decision making, the literature on algorithmic fairness has shown that predictive models can also perpetuate or exacerbate societal biases, leading to potentially unfair outcomes \citep{barocas2023fairness}.\looseness=-1




Recent work argues that in such high-stakes settings, firms building data-driven decision-making systems should proactively search for ``less discriminatory algorithms'' (LDAs) \citep{black_less_2024, gillis_operationalizing_2024, caro2024modernizing}, or predictive models with equal overall performance but less disparate impact across legally protected groups.\footnote{In the United States, disparate impact in these sectors is typically operationalized as the difference in selection rates across groups (e.g., differences in the hiring, lending, or leasing rates across racial, gender, or age groups).} \proposedchange{These arguments draw on the legal doctrine that has developed over the past half century around disparate impact, which stipulates that firms may face liability if they refuse to adopt a less discriminatory alternative that could achieve their business goals as effectively as some baseline decision-making policy.}

In support of their argument is the empirical finding that models optimized for accuracy can vary substantially with respect to other performance measures (like {disparate impact}), \textit{even if the training procedure used is exactly the same} \citep{marx_predictive_2020,d2022underspecification,rudin2024amazing,black2022model}. 
This is because training processes are almost always non-deterministic;
the subset of data used to train a model, the batch ordering in stochastic gradient descent, the set of features included as inputs, and any number of other aspects of a training algorithm are random.
A firm might thus hope to sample a large set of models with comparable predictive performance and select the one with minimal disparate impact.

Scholars, advocates, and regulators have argued that firms are well-positioned to search for LDAs because they oversee model training \citep{black_less_2024,FinRegLab_2023_ExplainabilityFairness,blower_cfpb_2023}. They have further argued that firms ought to take certain minimal steps to perform such searches, given that reductions in discrimination are sometimes achievable ``for free'' (i.e., without sacrificing accuracy) \citep{islam_can_2021,rodolfa_empirical_2021}. 
%


Others have been skeptical of the promise of LDAs, questioning whether they can really yield meaningful reductions in disparate impact and raising concerns about the lengths to which a firm must go to demonstrate a good-faith effort \citep{fools-gold-ii,scherer2019applying}. 
At the heart of this skepticism is the sense that a search for LDAs could potentially go on forever, given that additional searching might uncover an even less discriminatory alternative than what has been discovered already. Many modern machine learning optimization problems involve rich, infinite function classes and are highly non-convex, so it is often unreasonable to hope for global optima. (In some simpler machine learning methods or settings, however, it may be possible to find the \textit{least} discriminatory alternative \citep{gillis_operationalizing_2024}.)
As one financial services blog put it: ``no constraints or limits on this search have been proposed --- and it is unclear how much resources, time, and effort are expected in searching for these potential LDAs'' \citep{PaceAnalytics2022_SixUnansweredFairLendingQuestions}.
%

\proposedchange{ 
The relevant law, case law and related doctrine, and regulatory rules and guidance do not offer a consistent or precise answer to the question of how much---if any---costs firms are expected to incur in searching for or adopting an LDA. While the case law suggests that courts do not expect firms to shoulder \emph{significant} costs, it also suggests that firms cannot dismiss LDAs out-of-hand if they involve \emph{any} costs \cite{black_less_2024}. Similarly, while the Consumer Financial Protection Bureau, in recent years, called for firms to engage in the search for LDAs \cite{blower_cfpb_2023}, the Department of Housing and Urban Development has, at other times, explicitly stressed the idea that firms need not incur ``materially greater costs'' or ``other material burdens'' in adopting an LDA \citep{us_department_of_housing_and_urban_development_huds_2020}, only to later say that ``a proposed less discriminatory alternative [will not] fail simply because there will be some amount of increased cost associated with the alternative policy'' \citep{us_department_of_housing_and_urban_development_reinstatement_2023}. Taken together, these decisions and statements seem to imply a middle ground between two extremes that constitutes a ``good faith'' or ``reasonable'' investment. Yet it remains an open question how to operationalize this idea, especially when the cost at issue is the expense involved in \emph{developing} a model, not just \emph{adopting} one that might have lower performance.}

In this paper, we develop a framework for formalizing the concept of a sufficient search for less discriminatory algorithms in the model development process and a procedure for establishing that a firm has performed a sufficient search. 
%
%
%
%

\paragraph*{Our contributions.}

We develop statistical tools to quantify the value of an LDA search under model multiplicity.
We formalize LDA search as an optimal stopping problem, wherein a firm wants to continue training models as long as the marginal gain  from doing so (in disparate impact reduction) is sufficiently large.
This is a natural way to model the search process: the firm will not know about the models it may train until it actually trains them; thus, they should sequentially train new models until the returns from doing so are no longer worth the model training costs.
Our primary contribution is an optimal stopping algorithm (\Cref{alg:infinite}) and theorem (\Cref{thm:finite}) to quantify and bound the value of continuing a search for LDAs.
Our theorem provides a high-probability upper bound on the marginal value of training additional models, allowing a firm to stop an LDA search when its value is sufficiently low.
Thus, our methods also provide a \textit{certificate} of the limited benefits to a continued search, allowing the firm to demonstrate to a third party (e.g.,~a court or internal compliance team) that it has conducted a reasonable search.
%
Our framework allows for the firm to impose knowledge about data and model distributions in order to further refine our algorithm's guarantees.
Under stronger assumptions, we establish correspondingly stronger upper bounds on the marginal value of training additional models.


At a technical level, our algorithm establishes general high-probability guarantees for marginal returns of additional samples when sampling from an unknown distribution.
%
We draw on recent results on \textit{anytime-valid inference}, which allow us to adaptively stop training models while maintaining statistical validity.
In particular, we develop a novel and asymptotically near-optimal sequence upper-bounding the probability of improving upon a running best sample drawn iid from any distribution, which may be of independent interest.

We also evaluate our algorithm empirically on a number of publicly available datasets related to credit and housing.
We randomly retrain models across standard machine learning methods and measure the stopping time of our algorithm against the optimal full-information stopping time.
%
We find significant heterogeneity in the true, full-information marginal returns to retraining, and in the performance of the algorithm relative to this idealized benchmark. 
In all of our dataset/model class pairs, training a small number of models was enough: after about 60 models, the marginal gains in performance drop into the hundredths of a percent per new model; in several, the marginal gains were extremely small after less than 10 models trained.


\proposedchange{Our proposed procedure is just one piece of a larger and more complex set of steps that a firm might take to search for a less discriminatory algorithm.
%
In real cases, debates about the existence of a less discriminatory algorithm might cover a swath of both quantitative and qualitative considerations about the reasonableness of model assumptions, variables used, and so forth \citep{black_less_2024}. 
Our procedure may be helpful in a wholistic evaluation of the search:
For example, if our procedure establishes that it is not worth training beyond one hundred models, then a firm can stop at this time and allocate effort to a search in other parts of the ML pipeline. Moreover, if the firm trains a million models but our procedure establishes that one hundred was sufficient to achieve nearly the same benefits, a regulator might determine that search effort had been misallocated.
We also believe our framework may be useful in formalizing and reasoning about sufficiency of search along other dimensions, such as the amount of data collected or which features are recorded, in cases where such decisions can be formalized as samples from a distribution.}

%
%
There are many applications beyond the search for less discriminatory algorithms to which our framework might be applied. 
For example, one might use our approach to certify that they had done a sufficient search over the hyperparameter space for a given model class, by defining a (e.g., uniform) distribution over hyperparameters.
In this case, the search objective would be a higher performing model, not a less discriminatory one.
This might be especially fruitful in industry as a way to certify to a manager that a sufficient search was conducted for the best-performing model or as a way to certify to a journal or conference referee that a sufficient search for a strong baseline model was conducted.

Finally, our framework might be relevant to other problems of societal consequence.
For example, in computational redistricting, the goal is to partition a geographic area so each part satisfies certain properties like contiguity and population balance.
Popular approaches involve sampling from a distribution of valid districts, since direct optimization can be computationally intractable \citep{deford_recombination_2021,mccartan_sequential_2023}.
Our approach could be used to, e.g., repeatedly resample from a distribution of valid partitions to optimize for these properties, although if applying our approach to Markov chain Monte Carlo methods, one would have to account for interdependencies between samples.

\paragraph{Related work.} 
Our work is motivated by a literature on less discriminatory algorithms, model multiplicity and fairness/accuracy tradeoffs \citep{black_less_2024,black2022model,coston_characterizing_2021,rodolfa_empirical_2021,laufer_what_2025,gillis_operationalizing_2024,cen_audits_2025,fallah_statistical_2025, rudin2024amazing,dai2025intentional}. This literature surfaces the idea that there may be many highly accurate models, and that retraining models may yield predictors with different properties, especially with respect to fairness.
Our work addresses an important and unanswered question in this area: \textit{How do we certify the sufficiency of a search for a particular model retraining process?}

Our framework is closely related to a long literature in economics and computer science~\citep{degroot2004optimal,beyhaghi2024recent,lippman1976economics,bikhchandani1996optimal}, for which there are many classic applications beyond model retraining.
Our model is similar to the classic Pandora's Box problem \citep{weitzman1978optimal,kleinberg2016descending,beyhaghi2024recent}, in which the decision-maker pays a cost to sample from a \textit{known} distribution.
However, our work is different in that we assume minimal knowledge of the distributions.
Also, rather than trying to maximize total utility of a search, we seek a high-probability guarantee on the marginal returns of drawing another sample.

\paragraph{Organization of the paper.} In \Cref{sec:model}, we formalize our setting, including the model retraining process and our goals. In \Cref{sec:theory}, we describe our theoretical results, including an algorithm for adaptively training models and a theorem with corresponding guarantees on the correctness of the stopping time.
In \Cref{sec:empirical}, we validate our method on real-world datasets for credit and housing.
Finally, in \Cref{sec:conclusion}, we discuss other applications of our technical approach and conclude.

\section{Setting and model} \label{sec:model}

At a high level, we study the problem of learning a predictive machine learning model from a finite dataset.
The firm's utility for a predictor is determined by its average performance on a loss function over the population distribution.
In the search for LDAs, the loss function might be the difference in selection rates of a protected group versus that of a reference group.

The firm seeks to take advantage of model multiplicity to reduce this loss by sampling multiple high performing models and selecting the least discriminatory among them (i.e., the one that minimizes disparate impact).
Our target is to design a procedure which determines when a sufficient search has been conducted during the re-training process.

We assume the model trainer pre-specifies (1) a cost for sampling an additional model by repeating a randomized training procedure and (2) a utility for a unit improvement to disparate impact.
The ratio of these quantities specifies a target threshold for determining whether the marginal benefit of retraining models is worth the cost: if the expected benefit from training a new model is above the threshold, the model trainer should do so, and if it is below the threshold, the trainer should terminate the retraining procedure and deploy the best model seen so far.
In the remainder of this section, we formalize this setting and define notation.

\paragraph{Data and utility.} We will assume the existence of an unknown population distribution $\mc D$ from which the firm has sampled an iid dataset $D$ of size $n$, consisting of labeled data pairs $(x, y) \in \cX \times \cY$. 
The firm will deploy a predictor $h \; : \; \cX \to \cY$.
In cases where the predictor determines outcomes (like offers of employment, credit, or housing), $\cY$ will be binary, where $1$ is the positive outcome.
The firm's utility will be defined as
\begin{align*}
    Q(h)
    &\triangleq \EE{(x, y) \sim \mc D}{\ell(h(x), y, x)}
\end{align*}
for $\ell$ given and $\mathrm{im}(Q) \subseteq [0, 1]$.\footnote{This is without loss of generality: any bounded loss function can be rescaled so the loss is on $[0,1]$. Our proposed methods therefore work for loss functions beyond disparate impact; for example, a firm could minimize a weighted combination of disparate impact and error rate instead of disparate impact alone.}
If the goal is to reduce disparity in selection rates with respect to a group indicator $g(x) \in \{0, 1\}$ as in the search for LDAs, $\ell$ would be written as $\ell^{\mathrm{DI}}(a, y, x) = \paren{({1 - g(x)})/{P(g(x) = 0)} - {g(x)}/{P(g(x) = 1)}} a$,
which is $a/P(g(x)=0)$ if $g(x) = 0$ and $- a/P(g(x)=1)$ if $g(x) = 1$.
Then, the expected selection rate disparity of the model would be given by
    $Q^{\mathrm{DI}}(h) = \E[h(X) \given g(X) = 0] - \E[h(X) \given g(X) = 1],$
i.e., the difference between the selection rate for the reference group and the selection rate for the protected group.
The loss is bounded in $[0, 1]$ if the selection rate for the protected group ($X$ for which $g(X) = 1$) is never greater than that for the reference group ($X$ for which $g(X) = 0$). This is reasonable because discrimination against the protected group is not a concern if their selection rate is higher than the reference group.\footnote{If selection rate disparity is a concern for both groups (i.e., both groups are protected), this loss could alternately represent the absolute value of the difference between selection rates between groups.}
%
However, our results are not solely relevant to $\ell$ as the selection rate disparity: our results hold for any outcome space $\cY$ and loss function $\ell$ as long as the range of $Q$ is bounded in $[0,1]$.

%
%

The model trainer cannot observe their true utility.
Instead, we will assume they have access to a finite sample of data on which they will evaluate their model.
The empirical performance will be defined for a fixed dataset $S$, as 
\begin{align*}
    \hat{Q}(h; S)
    &\triangleq \frac{1}{|S|} \sum_{i \in S} \ell(h(x_i), y_i, x_i).
\end{align*}

\paragraph{Model distribution.} The model trainer will have a randomized training procedure $\mc A$ that takes in a dataset $D$ and returns a model $h$.
There are no assumptions on the procedure $\mc A(D)$, except that it is fixed in advance and returns a model iid conditional on the data $D$.

While we assume the model trainer has a fixed dataset $D$, we do not necessarily assume that all models are trained on the same training sample.
Instead, the data may be partitioned into subsets $\Dtrain$ and $\Dtest$, where $h = \mc A(D)$ depends only on the training subset $\Dtrain$ and not on the remaining data $\Dtest = D \backslash \Dtrain$.
Additionally, $\cA$ is not restricted to produce models from any particular model class or setting of hyperparameters---it does not need to be a standard model training process for a fixed model class.
For example, $\cA$ might first randomly decide between multiple algorithms (which themselves might be randomized), like random forests or neural networks.
Alternately, $\cA$ might sample from a given distribution over hyperparameters.

We will analyze the setting in which a model trainer trains a sequence of models $h_1, h_2, \dots$ by sampling iid, conditional on $D$, from $\cA(D)$.
Let $\Dtrain_1, \Dtrain_2, \dots$ be the sequence of training splits and $\Dtest_1, \Dtest_2, \dots$ be the sequence of test splits.
(Recall $\Dtrain_t \cup \Dtest_t = D$ for all $t$, so train and test splits for different steps $t$ will have shared data.)
For brevity, we will write the true and empirical loss of the $t$-th model as
\begin{align*}
    Q_t \defeq Q(h_t), \quad\quad \text{and} \quad\quad
    \hat Q_t \defeq \hat Q(h_t; \Dtest_t).
\end{align*}
We will denote by $P$ the distribution of the infinite sequence $Q_1, Q_2, \dots$. (When just considering the first $t$ entries of this sequence, we will imagine throwing the rest away so as to not introduce new notation.)
We will denote by $\hat P$ the distribution of the infinite sequence $\hat Q_1, \hat Q_2, \dots$ similar to $P$, and assume that $P$ and $\hat P$ are defined on the same space.
All distributions and probabilities throughout this work are taken conditional on $D$, since we imagine there is one fixed dataset used for training and evaluation.
Note that $P$ and $\hat P$ are supported on (a subset of) $[0, 1]^\infty$, since $Q_t, \hat Q_t \in [0, 1]$ by assumption.
Also, note that $\{ Q_{t} \}_{t=1}^\infty$ are iid, conditional on $D$.
Let $P_0$ be the marginal distribution of any $Q_t$.
Similarly, $\{ \hat Q_{t} \}_{t=1}^\infty$ are iid  conditional on $D$ and we will denote the marginal distribution of any $\hat Q_t$ by $\hat P_0$.
Finally, let $\bbP$ be the joint probability distribution over the pairs $(Q_1, \hat Q_1), \dots$ and let $\bbP_0$ be the marginal distribution over any $(Q_t, \hat Q_t)$.
%

We will analyze the model with the best performance on the test split, after the trainer concludes training. 
Formally, for given $t$, let $i_t$ be the model with the lowest empirical disparate impact up to the $t$-th model:
    $i_t = \arg \min_{i \in [t]} \hat Q_i$.
We will analyze the case where, after the model trainer trains $\tau$ models, they select and deploy $h_{i_\tau}$.
The true and empirical disparate impact of the \textit{selected} model after training $t$ models will be denoted
\begin{align*}
    U_t \defeq Q_{i_t}, \quad\quad \text{and} \quad\quad
    \hat U_t \defeq \hat Q_{i_t}.
\end{align*}
%

In the context of an LDA search, we assume the models sampled from $\cA(D)$ are all \textit{deployable}, in the sense that a sample from $\mc A(D)$ meets the business needs of the firm.
If this is not true for some model training process, rejection sampling can be used to continue retraining until a deployable one is found.
In practice, this may be accomplished by, for example, setting an accuracy threshold and letting $\mc A(D)$ be samples from the model training distribution, conditional on sufficient accuracy.\footnote{In other settings, $\ell$ might represent accuracy itself, in which case the search would be for more accurate models. However, our motivation for this work is clarifying the debate around LDAs. The model multiplicity literature argues that models optimized for accuracy will have similar accuracy but perhaps differences in other properties \citep{black2022model,rodolfa_empirical_2021}.} 
\proposedchange{This approach---setting an accuracy threshold and searching for a model with lower disparate impact---reflects the fact that the legal framework focuses on a search for alternatives among procedures that meet business needs. In other settings, beyond the LDA framework, $\ell$ could encode some fairness-accuracy trade-off via a weighted combination of different objectives.}

\paragraph*{Certifying a sufficient search.}
For given cost of training a single model $c$ and utility for a unit improvement to disparate impact $b$, the model trainer is justified in terminating a search after training $\tau$ models if 
    $b \cdot \E_{\bbP_0}[{U_\tau - U_{\tau+1} \given \hat U_\tau}] \le c,$
i.e., the expected marginal benefit of training an additional model, given the observed best model so far, does not outweigh the cost.
Equivalently, we will write 
\begin{align*}
    \E_{\bbP}[U_\tau - U_{\tau+1} \; | \; \hat U_\tau ] \le \gamma
    \numberthis \label{eq:true-stop-condition}
\end{align*}
where we define $\gamma \triangleq c / b$.
%
%
%
Our definition requires the model trainer to continue sampling models as long as the expected benefits outweigh the cost.
But our information is limited in two ways:
    First, we do not know $P$.
    Second, we can only observe noisy estimates of $Q_t$ due to our finite data sample.
Thus, we can only hope to upper bound the left-hand expression of~\Cref{eq:true-stop-condition}, with high probability over $\tau$, given this uncertainty.

Finally, we note that, while our guarantee is written in terms of the marginal benefits of training a single additional model, our results hold for the marginal benefit of training $k \geq 1$ additional models.
This is because the marginal benefits of retraining are monotonically non-increasing in each additional model, while the costs remain linear.
Thus, in our framework, if there exists some $k > 0$ for which training $k$ more models is worth the marginal cost, training one more model is worth the marginal cost.

\section{Adaptive Stopping for Repeated Model Retraining} \label{sec:theory}

Our main theoretical contribution is an adaptive algorithm (\Cref{alg:infinite}) and accompanying theoretical result (\cref{thm:finite}).
The algorithm gives a procedure for training models until a stopping condition is met.
The theorem establishes that, when the algorithm halts, the marginal benefits of retraining can be concluded to be no longer worth the costs.
We also establish that the algorithm always halts at some finite time that depends on $\gamma$ and gives a data-independent upper bound on the number of models that need to be trained.

Our plan for the section is as follows.
To build intuition, in \Cref{sub:fullinfo,sub:infinite}, we start with analyses of simpler settings. 
In \Cref{sub:fullinfo}, the distribution of model performance is known, and observations of performance are observed exactly as if they were evaluated on infinite data (i.e.,~$\hat Q_t = Q_t$ for all $t$).
In this regime, the stopping problem is trivial and can be described by a threshold on draws from the model peformance distribution.
Next, in \Cref{sub:infinite}, we relax the first condition and do not assume full knowledge of the model performance distribution.
We outline how different conditions on the model performance distribution yield different bounds, and our method allows decision-makers to input assumptions suitable to their context.
Then, in \Cref{sub:finite}, we handle the additional uncertainty from evaluations on finite data.
To do so, we introduce a natural assumption on the relationship between observed and true model performance.
Finally, in \Cref{sub:est-mrl}, we consider the case in which estimation of a property of the model loss distribution can be leveraged to produce tighter bounds on marginal benefits of model retraining.
All proofs are deferred to \Cref{sec:proofs}.

\subsection{The full-information regime} \label{sub:fullinfo}
We first consider the simplest case, when both the distribution $P$ is known and the population values of $Q_t$ are exactly observed.
For any $t$, note that
   $U_t - U_{t+1}
    =
    (U_t -Q_{t+1}) \cdot \ind\b{U_t > Q_{t+1}}$.
Thus, if the performance of the best model so far is $u$, the expected marginal gain of a new sample is
\begin{align*}
    g(u)
    &\triangleq
    \E_{Q \sim P_0}\sqparen{(u - Q) \cdot \ind\b{u > Q}}. 
    \numberthis \label{eq:stopping-lhs-simplified}
\end{align*}
Observe that $g$ is weakly monotonically increasing, and $g(0) = 0$.
Therefore, there is some threshold $u_P^*$ at which the marginal gain drops below $\gamma$. Define this threshold as follows:
\begin{align*}
    u_{P}^*
    &\triangleq \sup_{u \in [0, 1]} \{u : g(u) \le \gamma\}.
\end{align*}
Thus, our stopping time $\tau$ satisfies the desired guarantee~\Cref{eq:true-stop-condition} if and only if
\begin{equation}
    \E_{P_0}\sqparen{U_\tau - U_{\tau+1} \given U_\tau} \le \gamma
    \Longleftrightarrow
    g(U_\tau) \le \gamma
    \Longleftrightarrow U_\tau \le u_P^*.
    \label{eq:guarantee-equivalence}
\end{equation}
This immediately yields a stopping condition: compute 
$u_P^*$ and sample until a value less than $u_P^*$ is observed.
The stopping time $\tau$ in this case is geometrically distributed, since each sample is less than $u^*_P$ with probability $P_0(u^*_P \ge U_{\tau+1})$, and so the expected stopping time is $1/P_0(u^*_P \ge U_{\tau+1})$.

\subsection{The infinite-data regime} \label{sub:infinite}

Without knowledge of $P$, the developer cannot compute $u_P^*$.
We next consider the case where $P$ is unknown, but we can perfectly observe $Q_t$ for all $t$.
%
Because of our uncertainty about $P$, we cannot always guarantee \Cref{eq:true-stop-condition} for finite $\tau$: there is always a chance that the sequence $\{ Q_s \}_{s=1}^t$ observed so far have been abnormally large (i.e., an especially unlucky sequence), so that the expected marginal gain of a new sample is greater than $\gamma$.
The best we can do is ensure that it holds \textit{with high probability}, over the randomness of $\{Q_t\}_{t=1}^\infty$.
That is, for a pre-specified $\delta \in (0, 1)$, we want
\begin{equation}
    P(\E_{P_0}\sqparen{U_\tau - U_{\tau+1} \given U_\tau} \le \gamma) = P(g(U_\tau) \le \gamma) \ge 1-\delta,
    \label{eq:hp-guarantee}
\end{equation}
where the expectation is over $U_{\tau+1}$ marginally and the probability is over all $t$ jointly.
Our goal is thus to provide an anytime-valid upper bound on $\{g(U_t)\}_{t=1}^\infty$. 
%
That is, suppose we had a sequence $\{\bar g_t(U_t)\}_{t=1}^\infty$ such that
\begin{align*}
    P( \exists t \in \N \; : \;~ g(U_t) > \bar g_t(U_t) ) \le \delta.
\end{align*}
%
Then, it suffices to stop sampling at $\tau$ such that $\bar g_\tau(U_\tau) \le \gamma$, since this immediately provides a high-probability bound on $g(U_\tau)$.
%


We have thus reduced our stopping problem to maintaining an anytime-valid upper bound for $g(U_t)$.
Our next step is to actually construct such a bound.
To do so, we decompose $g(\cdot)$ into two terms: One which captures the probability of observing a strictly better sample, and another which captures the expected improvement \textit{conditional} on observing a strictly better sample. Observe
\begin{align*}
    g(u)
    &= 
    \E_{Q \sim P_0}\sqparen{u - Q \given u > Q}
    P_0(u > Q) 
    = \mu(u) p(u),
\end{align*}
where we define, for a draw of $Q$ iid from $P_0$,
    $\mu(u)
    \triangleq
    \E_{P_0}\sqparen{u - Q \given u > Q}$ and $p(u) \triangleq P_0(u > Q).$
We will call $\mu$ the \textit{conditional expected improvement} (CEI)\footnote{This concept is closely related to that of the \textit{mean residual life} of a random variable, for which there is a rich literature. See, e.g., \citet{almudevar_estimation_2020}.} and $p$ the \textit{improvement probability}.
%
It suffices to upper bound each of these separately and then combine them. 

\paragraph*{Bounding $\mu$.}
We first formalize our goal for bounds on $\mu$.
We will then provide explicit bounds under a variety of potential assumptions on the distribution.
The definition is written for a generic distribution $\cP$ since we will reuse this definition later in the finite-data case.
\begin{definition}[$\bar \mu $-Bounded CEI for $\cP$] \label{def:bounded-mrl}
    $\bar \mu: [0, 1] \to [0, 1]$ is a CEI bound for distribution $\cP$ if
    $$\E_{Q \sim \cP}\sqparen{u - Q \given u > Q} \leq \bar \mu(u)$$
    for all $u \in [0, 1]$, almost surely.
\end{definition}

We can use the fact that $\cP$ is supported on $[0, 1]$ almost surely to derive an immediate bound satisfying \Cref{def:bounded-mrl}:
$\E_{Q \sim \cP}[u - Q \given u > Q] \le u$, since $Q \ge 0$.
Thus, $\bar \mu^{\text{universal}}(u) \triangleq u$ satisfies \Cref{def:bounded-mrl}, making it a valid upper bound for $\mu$.
%

\begin{table}[t]
    \centering
    \begin{tabularx}{\linewidth}{p{0.3\textwidth} p{0.28\textwidth} X}
        \toprule
        Assumption & Interpretation & $\bar \mu$ \\
        \midrule
        No assumption &
        Applies to any distribution &
        $\bar \mu^{\text{universal}}(u) \triangleq u$ \\
        \assumplabel{as:monotone}
        $\exists a > 0$ s.t. $f_{P_0}(x)$ is increasing for $x \le a$ &
        $P_0$ has a sub-uniform left tail &
        $\bar \mu^{\text{mono}}(u) \triangleq \begin{cases}
            u & u > a \\
            \frac{u}{2} & u \le a
        \end{cases}$ \\
        \assumplabel{as:exp}
        $\exists a > 0$ s.t. $\mu(u)$ is increasing for $u \le a$ &
        $P_0$ has an exponential or sharper left tail &
        $\bar \mu^{\text{exp}} \triangleq \begin{cases}
            u & u > a \\
            \min(\mu(a), u) & u \le a
        \end{cases}$
        \\
        \assumplabel{as:bounded}
        $\exists a > 0$ s.t. $P_0(Q < a) = 0$ &
        No model has disparate impact lower than $a$ &
        $\bar \mu^{\text{bounded}} \triangleq u-a$
        \\
        \bottomrule
    \end{tabularx}
    \caption{Assumptions on $P_0$ and corresponding bounds $\bar \mu$.}
    \label{tab:assumptions}
\end{table}

This bound is quite conservative, since it bounds \textit{expected improvement} by \textit{maximum possible improvement.}
In what follows, we provide a series of assumptions on $P_0$ under which we can derive tighter bounds $\bar \mu$ satisfying \Cref{def:bounded-mrl}. These are summarized in \Cref{tab:assumptions}.
%

%
%
Our universal bound is conservative because it is tight only when $P_0$ places all mass to the left of $u$ at $0$.
Intuitively, this means that any model that performs better than $u$ is perfect (i.e.,~has $0$ disparate impact).
In practice, it would be surprising for this to be the case.
Instead, we might expect a continuum of models, where better models are more rare than worse models.
We can formalize this intuition by imposing assumptions on $P_0$, each of which leads to a different bound $\bar \mu$.

First, consider the case when $P_0$ is continuous, and there exists $a \in (0, 1)$ such that $f_{P_0}(x)$ is non-decreasing in $x$ for all $x \le a$ (i.e., $P_0$, at worst, has a uniform-like left tail).
This captures our intuition that better models are at least as rare as worse models.
By this assumption, $\mu(u) \leq \int_{0}^u x/u \dif x = u/2$.
Thus, we can define $\mumono(u)$ as $u/2$ if $u \le a$ and $u$ otherwise, giving us an upper-bound on $\mu$.

We could strengthen this assumption by requiring $P_0$ to put strictly decreasing mass on better models.
In other words, we could assume that the left tail of $P_0$ drops sufficiently quickly.
If the left tail of $P_0$ drops at an exponential rate, then there exists some $b > 0$ such that $\mu(u) = b$ for sufficiently small $u$.
We generalize this property to capture distributions with tails that fall off at least exponentially quickly.
Suppose there exists some $a \in (0, 1)$ such that the $\mu(u)$ is weakly increasing for all $u \leq a$, then the bound $\bar \mu^{\text{exp}}(u) = \min (\mu(a), u)$ if $u \leq a$ and $u$ otherwise satisfies \Cref{def:bounded-mrl}.

Finally, we consider the case where the best possible model is still bounded away from $0$, meaning there is no model with $0$ disparate impact.
Formally, suppose $P_0(Q < a) = 0$ for some $a \in (0, 1]$.
Then we can define $\mubound_t(u) \triangleq u - a.$
In practice, using these bounds requires the developer to believe, \textit{a priori} that one of these assumptions holds, and importantly, the particular value of the relevant parameter $a$.
In \Cref{sub:est-mrl}, we discuss how, under some assumptions, they can infer parameters of a bound from data.
%

\paragraph*{Bounding $p$.}

Having provided bounds for $\mu$, we next turn to the problem of bounding $p$.
The following lemma yields a general anytime-valid high probability upper bound for the probability of observing a new minimum in a sequence of iid random variables. We state the lemma for a general sequence of random variables because it may be of independent interest.
\proposedchange{Informally, the lemma establishes a sequence $\bar p_t$, depending on a confidence budget $\alpha$, such that with probability at least $1-\alpha$, $\bar p_t$ is greater than the minimum value seen so far among a sequence of iid random variables. The bound is \textit{anytime-valid} because it holds for all time steps $t$, rather than a traditional bound, which would hold for some $t$ fixed in advance.}
An asymptotically near-optimal (but more complex) sequence can be found in \Cref{thm:always-valid-min-iterated-log}.
\begin{restatable}{lemma}{avmin}
    \label{thm:always-valid-min}
    Let $\{X_t\}_{t=1}^\infty$ be a sequence of iid random variables distributed according to a law $\cP_0$. Let $\cP \defeq \cP_0^\infty$ be their joint distribution.
    Let $Y_t$ be the minimum of these variables up to time $t$, i.e., $Y_t \defeq \min_{s \le t} X_s$.
    For any $\alpha \in (0, 1)$, define an upper bound
    \begin{align*}
        \bar p_t(\alpha)
        &= \begin{cases}
            1 - e^{-1/\alpha} & \text{ if } t = 1\\
            1- \paren{\frac{(t-1)}{\alpha} + 1}^{-1/(t-1)} & \text{otherwise.} 
        \end{cases}
    \end{align*}
    Then, \proposedchange{the probability over $\cP_0$ that $X_{t+1}$ is less than $Y_t$ at any time $t$ is bounded by $\bar p_t(\alpha)$ with probability $1-\alpha$ over $\cP$. Formally,}
    \begin{align*}
        \cP(\exists t \in \N \; : \; ~ \cP_0(X_{t+1} < Y_t \given Y_t) > \bar p_t(\alpha)) \le \alpha.
    \end{align*}
\end{restatable}



\proposedchange{We note that the high probability bound is \textit{conditional on $Y_t$}, reflecting the fact that a stopping procedure may depend on the minimum seen so far. This stands in contrast with an unconditional bound, which could not depend on the data seen so far and so would not be relevant to a decision-maker deciding whether to stop at time $t$ or not, having seen the data up to time $t$.}
\Cref{thm:always-valid-min} yields an immediate anytime-valid upper bound on $\{p(U_t)\}$:
\begin{equation}
    \label{eq:pt-bound}
    P(\exists t \in \N \; : \; ~ p(U_t) > \bar p_t(\delta)) \le \delta.
\end{equation}

\paragraph*{Combining bounds.}

Our algorithm simply combines our bounds on $\mu$ and $p$ to maintain an anytime-valid upper bound on the marginal gain, given by $\bar \mu(U_t) \cdot \bar p_t(\delta)$.
Formally, our algorithm simply terminates at the first $\tau$ such that
    $\bar \mu_\tau(U_\tau) \cdot \bar p_\tau(\delta) \le \gamma.$
Moreover, $\tau$ is guaranteed to be finite because $\bar \mu_t(\cdot) \le 1$ for all $t$, and $\lim_{t \to \infty} \bar p_t(\delta) = 0$.
A data-independent upper bound on the maximum possible number of models trained by our algorithm can thus be directly computed from $\delta$ and $\gamma$ by finding the smallest $t$ such that $p_t(\delta) < \gamma$.
We state the algorithm for a generic distribution $\cP$ given as input, rather than $P_0$, since we will reuse this algorithm in the finite-data regime.
\begin{algorithm}[H]
    \caption{LDA Search with Adaptive Stopping}
    \begin{algorithmic}[1]
    \REQUIRE ~\\
        An unknown model performance distribution $\cP$ from which to draw iid samples. \\
        Stopping threshold $\gamma$ and failure probability $\delta$. \\
        Optional: An almost-sure expected conditional improvement bound $\bar \mu$ satisfying \Cref{def:bounded-mrl}.
        If not provided, use $\bar \mu^{\mathrm{universal}}(u) = u$.
        
    \FOR{$t = 1, 2, \dots$}
        \STATE Draw a new sample $X_t \iidsim \cP$. 
        \STATE Define $\bar p_t$ as in \Cref{thm:always-valid-min}.
        \STATE Define $Y_t = \min_{s \leq t} X_t$
        \IF{$\bar \mu(Y_{t}) \cdot \bar p_t(\delta) < \gamma$} \label{line:utltc}
            \RETURN {$Y_t$}
        \ENDIF
    \ENDFOR
    \end{algorithmic}
    \label{alg:infinite}
    \end{algorithm}

    We now state the formal statistical guarantee for our infinite data setting. It is a special case of a more general theorem we prove, \Cref{thm:generic}. 
    \begin{proposition} \label{prop:infinite}
        For all $\gamma, \delta > 0$, \Cref{alg:infinite} run with $\cP = P_0$, $\gamma, \delta$ and any $\bar \mu$ that satisfies \Cref{def:bounded-mrl} for $P_0$ as input terminates at a stopping time $\tau \in \N$ such that
        \begin{align*}
            P( \E_{P}[U_\tau - U_{\tau + 1} \given U_\tau ] < \gamma ) \geq 1 - \delta.
        \end{align*}
    \end{proposition}
    Next, we generalize to the case where we have finite data.
    
\subsection{The finite-data regime} \label{sub:finite}


If we observe only finite data, we cannot perfectly observe each $Q_t$; instead, we observe $\hat Q_t$.
As before, we will seek to maintain an anytime-valid upper bound on the marginal gain.
We must take care to define the marginal gain appropriately---in particular, our goal is to bound the expected marginal gain with respect to the \textit{true} disparate impact ($Q_t$), given our observations of empirical disparate impact ($\hat Q_t$).
Formally, our goal is to show that, at stopping time $\tau$, 
\begin{align*}
    \E_{\bbP}[{U_\tau - U_{\tau + 1} \given \hat U_\tau }] \le \gamma.
\end{align*}
where the expectation is also conditional on $D$.
To do this, we need to establish a relation between the measurement error $U_t - \hat U_t$ at different points on the left tail of $\hat P_0$. 
We provide a natural assumption on the relationship between these quantities: the selection effect or regression-to-the-mean effect is, in expectation, non-decreasing in $t$. The assumption that regression-to-the-mean is at least constant is frequently supposed in the large literature on adjusting analysis for or estimating these effects \citep{stein1956inadmissibility,james1961estimation,sorensen1984estimation,andrews_inference_2024,zrnic_flexible_2024,fithian2014optimal}. 
Intuitively, this assumption holds for sub-Gaussian left tails where the selection effect should be linear in the gap between $\hat U_t$ and $\hat U_{t+1}$ and even for sub-exponential left tails where there should be constant regression to the mean in the gap between $\hat U_t$ and $\hat U_{t+1}$.
This assumption would not hold if some measurable set of values of $\hat U_t$ indicate that the model has low disparate impact, while models with $\hat U_{t+1} < \hat U_t$ have relatively high disparate impact.
{
\begin{assumption}[Non-decreasing selection effect] \label{assm:regression}
    It holds for all $t$ that
    \begin{align*}
        \E_{\bbP}[U_t - \hat U_t \given \hat U_t ] \geq \E_{\bbP}[U_{t + 1} - \hat U_{t + 1}\given \hat U_t ].
    \end{align*}
\end{assumption}
} 

Under \Cref{assm:regression}, we can apply \Cref{alg:infinite} on the sequence $\{ \hat U_t \}_{t=1}^\infty$ exactly the same as to how we applied it to $\{ U_t \}_{t=1}^\infty$ in the infinite data case.
This additional assumption is sufficient for the following theorem to hold, using only a minor modification to the argument applied in the infinite data case.
\begin{restatable}{theorem}{finite} \label{thm:finite}
    Under \Cref{assm:regression}, for all $\gamma > 0$ and $\delta > 0$, \Cref{alg:infinite} run with $\cP = \hat P_0$, $\gamma, \delta$ and any $\bar \mu$ that satisfies \Cref{def:bounded-mrl} for $\hat P_0$ terminates at a time $\tau \in \N$ such that
    \begin{align}
        \bbP(\E_{\bbP} [U_{\tau} - U_{\tau + 1} \given \hat U_\tau] \leq \gamma) \geq 1 - \delta. \label{eq:thm_finite}
    \end{align}
\end{restatable}

\subsection{Data-driven anytime-valid upper bounds on the conditional expected improvement.} \label{sub:est-mrl}

    We conclude this section with an analysis of the case in which $\bar \mu$ can be estimated from data.
    In particular, we analyze the case in which the following assumption, based on (\ref{as:exp}), is satisfied:
    \begin{assumption}[Non-decreasing CEI] \label{assm:mrl}
        For all constants $a, a'  \in \supp{\hat P_0}$ such that $a > a'$ and $a, a'< \median(\hat P_0)$, it holds with probability $1$ that
        \begin{align}
            {\E_{\hat P_0} [a - \hat Q_{t} \; | \; \hat Q_{t} < a] \geq \E_{\hat P_0}[a' - \hat Q_{t} \; | \; \hat Q_{t} < a'].} \label{eq:mrlassm}
        \end{align}
        simultaneously for all $t=1,2,\dots$.
    \end{assumption}

    In other words, better models have lower conditional expected improvement.
    The farther out into the tail of a distribution a model is, the smaller we expect future improvements to be.
    %
    %

    \Cref{assm:mrl} implies that the CEI at the left tail can be bounded by the CEI near the median.
    To make use of this assumption, we must develop a high-probability anytime-valid upper bound for this CEI.
    %
    %
    We formalize this in \Cref{alg:av3}.

    %
    At a high level, our approach is to first pick a quantile $q^{\mathrm{target}}$  of the distribution to estimate, and then estimate the CEI at a high probability lower bound on that quantile.
    Under \Cref{assm:mrl}, the CEI at this estimated quantile bounds the CEI in the tail of the distribution, giving us our high-probability bound $\bar \mu$.
    %
    %
    Unlike in \Cref{sub:infinite}, $\bar \mu$ is a high-probability upper bound and does not hold with probability 1.
    Thus, we must take a union bound to ensure that overall our guarantee holds with probability at least $1 - \delta$.

    The choice of $q^{\mathrm{target}}$ is arbitrary as long as it is below the median, but it must balance two factors.
    First, it must not be too close to zero. If it is, we will not have much data on which to estimate $\bar \mu$ and will thus have a loose upper bound (i.e., $\bar \mu$ will be large because it is estimated from a small amount of data).
    Second, the smaller $q^{\mathrm{target}}$ is, the smaller the CEI (based on \Cref{assm:regression}). That is, we will be estimating a smaller value of $E_{\hat P_0}[a - \hat Q_t \; | \; \hat Q_t < a]$ where $a$ is the $q^{\mathrm{target}}$-quantile of $\hat P_0$.
    This is because, by \Cref{assm:regression}, $E_{\hat P_0}[a - \hat Q_t \; | \; \hat Q_t < a]$ is non-decreasing in $a$.
    Thus, the first factor is about the difference between a high probability lower bound and the true quantity, and the second is about the magnitude of the true quantity.
    We'd like the combination of these two to be as small as possible.
    We leave exploration of the optimal choice of the quantile for future work.

    \begin{algorithm}[H]
    \caption{LDA Search with Adaptive Stopping and CEI Estimation}
    \begin{algorithmic}[1]
    \REQUIRE Stopping threshold $\gamma$ and failure probability $\delta$.
    \STATE Let $T_1 = \ceil{18 \log (3 / \delta))}$ and draw $T_1$ samples $\{ \hat Q_s \}_{s=1}^{T_1}$. 
    \STATE Compute the empirical quantile at level $1/3$:
        $$C \defeq \hat Q_{(\floor{T_1 / 3})}.$$
    \FOR{$t = 1, 2, \dots$}
        \STATE {Draw a new sample $\hat Q_t$ and compute $\hat U_t = \min_{s \leq t} \hat Q_s$.}
        \STATE Let $\bar p_t$ be defined as in \Cref{thm:always-valid-min}. Also, define
        \begin{align*}
            & {\Delta_t \defeq C - \hat Q_t} \\
            &{S_t \defeq \{i \leq t \; : \;  \indic{\Delta_i > 0} \}}\\
            &\bar \mu_t \defeq \begin{cases}
                \mueb_t(\{ \Delta_s \}_{s=1}^t, \delta/3, S_t) 
                & \text{if } {S_t} \neq \varnothing\\
                \hat U_t & \text{otherwise}
            \end{cases} 
        \end{align*}
        where $\mueb$ is defined as in \Cref{cor:mueb}.
        \IF{$\bar \mu_t \cdot \bar p_t (\delta / 3) < \gamma$} 
            \RETURN {$\hat U_t$}
        \ENDIF
    \ENDFOR
    \end{algorithmic}
    \label{alg:av3}
    \end{algorithm}

We are now ready to state our theorem.

\begin{restatable}{theorem}{estmrl} \label{thm:est-mrl}
    Under \Cref{assm:mrl,assm:regression}, for all $\gamma > 0$ and $\delta > 0$, \Cref{alg:av3} run with $\gamma$ and $\delta$ terminates at a time $\tau \in \N$ such that
    \begin{align}
        \bbP(\E_{\bbP} [U_{\tau} - U_{\tau + 1} \given \hat U_\tau] \leq \gamma) \geq 1 - \delta. \label{eq:est-mrl}
    \end{align}
\end{restatable}

\section{Empirical Analysis} \label{sec:empirical}

In this section, we evaluate our method on several datasets and machine learning methods. 
The datasets we use are Adult \citep{adult_2}, Folktables \citep{ding2021retiring}, and HMDA \citep{hmda}.
The first of these is a lending prediction task, the second is an employment prediction task and the third is a mortgage prediction task.
The methods we use are logistic regression, random forests and gradient boosted trees.
In \Cref{sec:further-empirical}, we also evaluate how well our approach composes with fairness-aware machine learning methods like Fairlearn \citet{weerts2023fairlearn}.
Full details of our data cleaning and feature/outcome selection are available in \Cref{sec:further-empirical} and in our code at \texttt{ \href{https://github.com/johnchrishays/lda}{https://github.com/johnchrishays/lda}}.

\paragraph{Evaluation procedure.} To evaluate our algorithm, we would ideally compare its performance against the full-information regime discussed in \Cref{sub:fullinfo}, where we perfectly observe the marginal benefit of sampling a new model.
This is in general not possible, since we know neither the true data distribution nor the true distribution of model disparate impacts.
Instead, we treat the finite dataset as a ``population distribution'' and subsample to produce training and testing datasets.
In particular, to create a population distribution $\cD$, we assign equal probability measure to each point in the dataset. 
To create a dataset $D$, we sample 3000 points iid from this distribution.
In practice, dataset sizes can vary considerably depending on the domain. For a fixed model class, fewer observations imply noisier estimation and a larger set of viable high-performing models, whereas more observations imply a more tightly identified neighborhood around the highest-performing models. Dataset size also affects training costs and can therefore be reflected in $\gamma$, the marginal cost of training an additional model.
%

Just as we cannot observe the population data distribution, we similarly cannot observe the population model training distribution.
That is, we cannot exactly compute probabilities or expectations with respect to $\cA(D)$, since model training processes are typically constituted by a series of possibly complex and opaque operations, and therefore 
do not lend themselves to closed form computations.
Thus, we use a similar technique to generate (and eventually subsample from) a large pool of $B = 5000$ trained models.
To generate the pool, each model $h_b, b \in [B]$ is trained on a partition $\Dtrain_b \subset D$ and then evaluated on the other partition $\Dtest_b \subset D$. This generates $\hat Q_b = \hat Q(h_t; \Dtest_b)$. 
Then, the population disparate impact $Q_b$ is computed by evaluating disparate impact over the whole population distribution (defined by the complete, non-subsampled dataset we started with).
This set of sample and population disparate impacts $\{ \hat Q_b, Q_b \}_{b \in [B]}$ is used to define $\bbP_0$.
We similarly record the sample and population accuracy of models trained.

To evaluate our algorithm, in each iteration we sample iid from $\bbP_0$ and run the algorithms described in \Cref{sec:theory}.
For a given $\bbP_0$, we run the algorithm many times (sampling from $\bbP_0$ up to a maximum of $T$ times for each run of the algorithm) and report results over all runs.
To compare these results against a ground truth, we can compute the expected marginal gain as in the full-information regime defined in \Cref{eq:stopping-lhs-simplified}, for given draws from $\bbP_0$.
All together, this procedure gives us a setup in which all true distributional quantities are known, allowing us to compare our method with a (semi-synthetic) ground truth.

Finally, in all of our results, to smooth out idiosyncrasies due to the realization of the sub-sampled dataset $D$, we repeat the whole procedure 45 times, generating a new dataset $D$ each repetition.
We note that all of the above setup is to facilitate a point of comparison (a ground truth against which to evaluate our algorithm). 
None of this setup would be necessary to deployment of our procedure, where any ground truth distributions are typically unknowable.
Further details on our data preparation, model training and comparison to the full-information regime are available in \Cref{sec:further-empirical}.

With the compute resources we used to train the models for this paper we estimate the compute costs as follows:
(Recall, for each population dataset and model combination, we train $B=5000$ models for each of $45$ dataset subsamples. This results in a total of about 4 million models trained in the final versions of our experiments.)
The total CPU-time for our final set of models trained was 137 days and 17 hours (a little more than three days in wall-clock time with parallelism), on machines with 515G of memory and 2.2GHz CPU frequency.
Servers with similar resources can be rented from AWS for about \$0.05 an hour per CPU at the time of this writing, which would result in a total training cost of about \$165, resulting in cost per model trained of approximately \$0.00004.
Of course, datasets of different size, different model classes, different hardware or different CPU-hour costs would yield different model training times and therefore different costs per model trained.

\paragraph{Analysis of model performance.} A precondition for there to be any benefit to model retraining is variation in disparate impact.
In \Cref{fig:acc_vs_srg}, we plot the population accuracy (vertical axis) and population disparate impact (horizontal axis) of our procedure across all models trained.
Brighter colors indicate higher densities of models, on a logarithmic scale.
Panel rows are machine learning methods and panel columns are datasets.

Observe that there is significant variation in disparate impact among models, in the sense that there is frequently about a $20\%$ spread between the highest and lowest disparate impact observed.
By contrast, for almost all model/dataset combinations, there is less than a five percent spread between the best and worst models trained.
The means and standard deviations of population disparate impact and accuracy is recorded in \Cref{fig:srg,fig:acc}, respectively.
Because of the relatively large variation in disparate impact and small variation in accuracy, there is not typically a large fairness/accuracy tradeoff within our data: finding the lowest disparate impact model would result in minimal accuracy compromises.
The greatest exception to this is random forests trained on HMDA, where taking the absolute lowest disparate impact model might yield about a 5\% decrease in accuracy. 

\paragraph{Analysis of \Cref{alg:infinite}.}
To evaluate \Cref{alg:infinite} in the finite-data regime, we track (1) the true marginal benefit of a continued search, and (2) the upper bound given by the algorithm. For a fixed $\gamma$, we are primarily interested in our algorithm's stopping time $\tau$, particularly as it compares to the time at which the true, full-information marginal gain drops below $\gamma$.
%

\begin{figure}
    \centering
    \includegraphics[width=\linewidth]{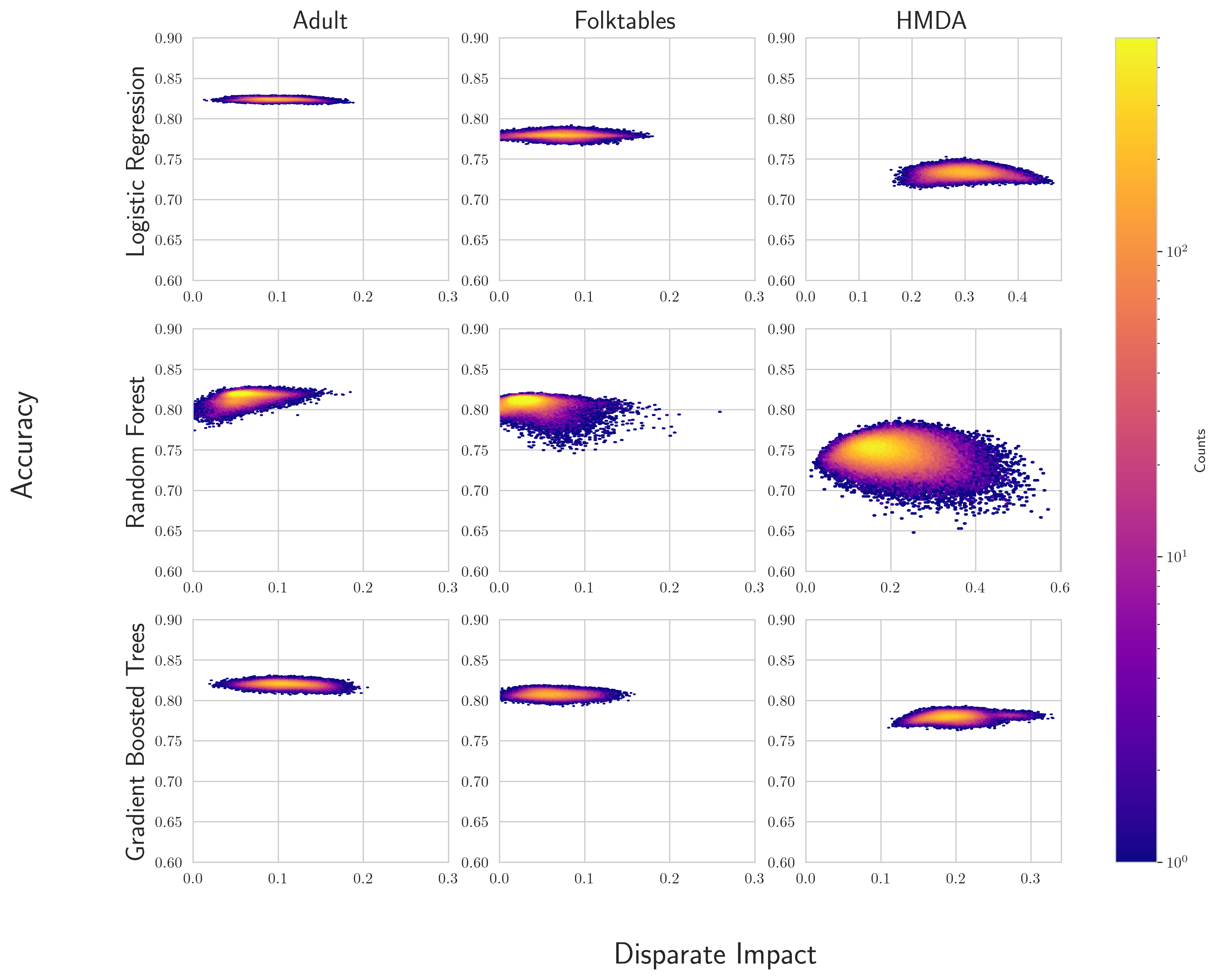}
    \caption{Heatmap of accuracy versus selection rate gap for each dataset and model class. Colors on the plot represent densities of trained models. Each model is trained on a small subset of the data and evaluated on the full population, so accuracy and disparate impact values are population quantities.
    For many datasets and methods, there is significantly more variation in disparate impact than there is in accuracy: the clusters of models spread horizontally more than they spread vertically. The fact that there are models with significant variation in disparate impact but similar accuracy supports the idea that simply retraining models can yield non-trivially less discriminatory algorithms.}
    \label{fig:acc_vs_srg}
    \vspace{-0.1cm}
\end{figure}

\begin{figure}
    \centering
    \includegraphics[width=\linewidth]{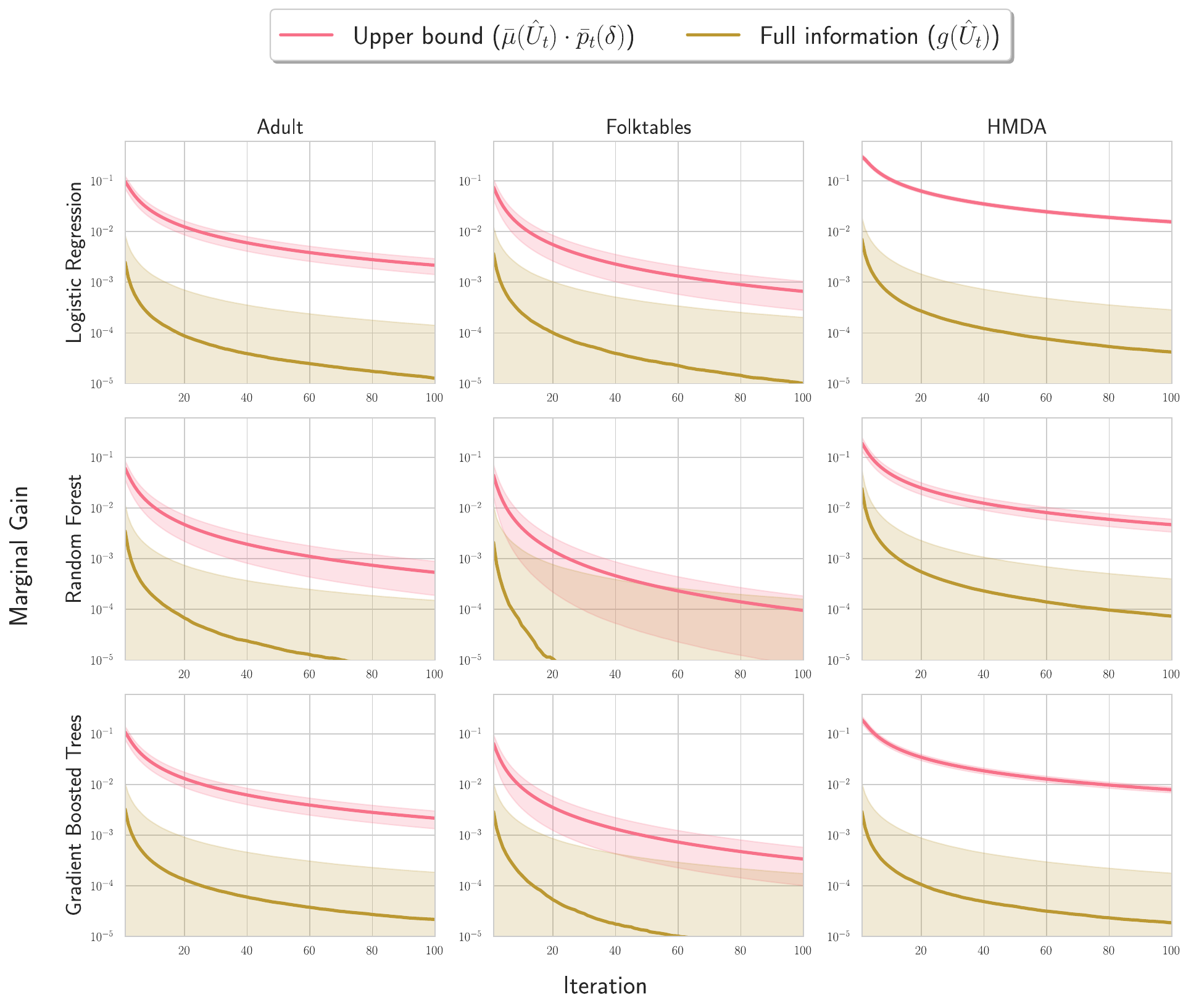}
    \caption{\Cref{alg:infinite} run on several datasets and models. Panel rows are ML methods and panel columns are datasets. In each panel, the horizontal axis is the iteration of the algorithm and the vertical axis is marginal gain. The pink line is our estimated upper bound $\bar \mu(\hat U_t) \bar p_t(0.05)$ and the brown line is the full-information marginal gain. For any $\gamma$, \Cref{alg:infinite} would stop when the pink line crosses the horizontal line at $\gamma$. \proposedchange{Note that the vertical axis is on a logarithmic scale, so a shallow downward sloping line reflects a polynomial decrease in the displayed quantity.}}
    \label{fig:cei}
    \vspace{-0.1cm}
\end{figure}

The results of running the algorithm many times for each $\bbP_0$ are visualized in \Cref{fig:cei}.
Iterations of the algorithm are on the horizontal axis and the marginal gain from resampling (using a logarithmic scale) is on the vertical axis.
The pink line is our upper bound $\bar \mu(\hat U_t) \bar p_t(\delta)$ for $\delta = 0.05$, setting $\bar \mu(\hat U_t) = \hat U_t$; we place no assumptions on the distribution $\hat P_0$.
The brown line is the ground truth $g(\hat U_t)$.
The shaded colored regions for each line show standard deviations over multiple runs of the dataset resampling, model training and algorithm. 

For any $\gamma$, \Cref{alg:infinite} would stop when the pink line drops below the horizontal line at $\gamma$.
Given full distributional information, a model trainer should stop the retraining process once the brown line drops below the horizontal line at $\gamma$.
Thus, for any fixed $\gamma$, the average number of iterations that the algorithm trained models past the stopping time given full information is the horizontal distance between the brown and pink lines.
Empirically, \Cref{alg:infinite} performs well in the sense that it ``overshoots'' the correct stopping time by tens of models in general, though it appears to perform worse for logistic regression on all datasets and all methods on HMDA.
Further assumptions (i.e.,~\ref{as:monotone}, \ref{as:exp}, \ref{as:bounded}) will likely yield tighter bounds.
\proposedchange{We note that each vertical axis is on a logarithmic scale. This means that even a shallow downward slope among the upper bound on marginal gain or full information reflects a polynomial decrease towards 0.}

We also implemented \Cref{alg:av3} in the same setting. The corresponding visualization of the results is displayed in \Cref{fig:cei_mrl}. 
We find that for these datasets and at this scale, the bounds from \Cref{alg:av3} are not better than those without the additional assumptions.
This is because we must account for uncertainty in estimating the conditional expected improvement, for which we must commensurately reduce our confidence budget when estimating the improvement probability.
The benefits from a tighter estimate of the CEI are not greater than the costs of a more conservative estimate of the improvement probability, so the overall bounds are slightly worse than when applying the trivial CEI bound.
The same may not necessarily be true for other contexts or datasets.

To supplement our results on popular machine learning methods, we run our algorithm on a popular fair machine learning framework Fairlearn \citep{weerts2023fairlearn}.
An advantage of Fairlearn is that it provides a wrapper around popular ML methods implemented in Python, allowing head-to-head comparisons between the Fairlearn and vanilla versions of a given method (i.e., logistic regression, random forest or gradient boosted tree).
We display the results in \Cref{fig:cei_fairlearn}.
The results for the Fairlearn versions of methods are comparable to those of the vanilla versions.

In practice, data and ML pipelines in industry for credit, housing and employment prediction may vary substantially among themselves and in comparison to our setup. 
For example, in typical credit model development processes, datasets may consist of hundreds of thousands to millions of observations. Each observation may consist of rich credit history data from credit bureaus and proprietary account history data, yielding tens to thousands of features.
See, e.g., \citet{FinRegLab_2023_Advancing} for an overview and demonstration of industry practices.
Since our models are trained on a much smaller amount of data, only publicly available data and fewer and perhaps less informative features, our empirical results may not match those that would be observed in typical industry contexts.
For example, we would expect that training on a larger sample would yield lower variance in the performance of models, perhaps leading to smaller marginal returns to retraining. Moreover, larger datasets are more expensive to train, yielding higher $\gamma$ in these larger data settings.
(For example, with our computational resources, it would not have been possible to create plots representing the distribution of model performance like those above on much larger datasets than the ones we used.)
Thus, our results should be taken as illustrative of our method but not necessarily representative of results that would be obtained in industry ML model development contexts.


\section{Discussion} \label{sec:conclusion}

Although recent work has proposed that firms should take steps to proactively search for less discriminatory algorithms, there are a number of open questions regarding both the gains to be expected from an LDA search and the resources required to conduct one.
In this paper we take one step towards developing the tooling firms would need to conduct a search. We put forward a method that allows firms to adaptively sample models that come from a particular loss distribution.
Our algorithm adaptively bounds the marginal gains of a continued search, allowing a firm to terminate the search when the gains are small and provide evidence that their search was sufficient.

We take as given the developer's cost of training models relative to their value of reducing disparate impact.
While determining how a firm might determine this cost is beyond the scope of this work, it is the subject of ongoing debate \citep{PaceAnalytics2022_SixUnansweredFairLendingQuestions,black_less_2024}. Case law suggests that courts may reject less discriminatory alternatives that impose too great a cost on defendants, but that they are also willing to endorse alternatives that are far from costless \cite{black_less_2024}. Our framework can help contribute to this debate in at least two ways.
First, because we provide anytime-valid bounds, we do not require that a firm pre-specify a cost.
Instead, model developers and compliance teams can iteratively develop models, consider the incremental gains, run separate experiments, and adaptively decide how to value those gains relative to development costs.
Second, given a search conducted by a firm, our framework allows us to ``back out'' a high-probability upper bound on the firm's estimated cost implied by their decision to stop the search.
That is, by observing a sequence of models sampled by a developer, we can draw conclusions about their implicit value for reducing disparate impact from their decision to terminate a search, and thereby facilitate a more informed debate about the reasonableness of the search.



A number of future directions related to this setting are open.
Our framework could be extended to handle adaptivity, where the performance of previous models informs training decisions for future models.
Extending our methods to account for adaptivity would likely require new technical tools.
It would also be interesting to explore \textit{label bias} within our framework, where outcomes are only observed in circumstances where the individual receives the positive outcome.
Additionally, in low-data settings, we would expect \textit{shrinkage} or \textit{selection effects} to be salient: the best-performing model in-sample could fare much worse out-of-sample, potentially admitting stronger guarantees.
Finally, our technical framework can be applied to general optimal stopping problems where high-probability guarantees are desirable.
For example, a developer or researcher using an LLM may choose the best of many randomly sampled prompts, and with our algorithm, they can certify that further exploration is unlikely to yield significant gains.
Applying our framework to other settings is a fruitful direction for future work.




\ifcslaw
    \bibliographystyle{iclr2026_conference}
\else
    \bibliographystyle{apalike}
\fi
\bibliography{refs,zotero}

\appendix

\section{Additional details on empirical analysis.} \label{sec:further-empirical}

\paragraph{Dataset preparation.} We use the pre-defined prediction targets, features and protected/reference groups given in the datasets. For Adult, the target is to predict whether income is above \$50k. For Folktables, we used the ACSEmployment task and filtered the data to individuals from Alabama from 2018. The prediction target is whether the individual is employed. For HMDA, we use the cleaned dataset given in \citet{cooper2024arbitrariness} for New York in 2017. The prediction target is whether the home mortgage was originated.
Across all datasets, the reference group is all individuals designated White and the protected group is all individuals not designated White.
The size specifications of our datasets, sub-sampling routines and runs were chosen to produce confident results and demonstrate a plausible approach to implementing the procedure described in this paper. 

\paragraph{Model training procedures.} We use default parameter settings for each of our ML methods, except for the following modifications:
For random forests, we fit ten estimators of depth no more than five. 
To explore how our method works when composed with other approaches for reducing disparate impact in model training, we additionally trained models using the Fairlearn python package \citep{weerts2023fairlearn}.
A convenience of Fairlearn is that it provides a wrapper around many standard Python ML model classes, including the ones we train in this paper.
Thus, we can compare Fairlearn versions of logistic regression, random forests and gradient-boosted trees with vanilla, non-Fairlearn versions.
For the Fairlearn versions of methods, we set the selection rate difference bound to 0.2.
The mean and standard deviation of the selection rate disparities and accuracy for each dataset and model class are in \Cref{fig:srg} and \Cref{fig:acc}, respectively.
After we generate the population and observed disparate impact for each model, we then run each algorithm in the paper by sampling $T$ models iid from the dataset of model performances to generate a model retraining trajectory. We then repeatedly resample to generate many trajectories.

\begin{figure}
    \centering
    \includegraphics[width=\linewidth]{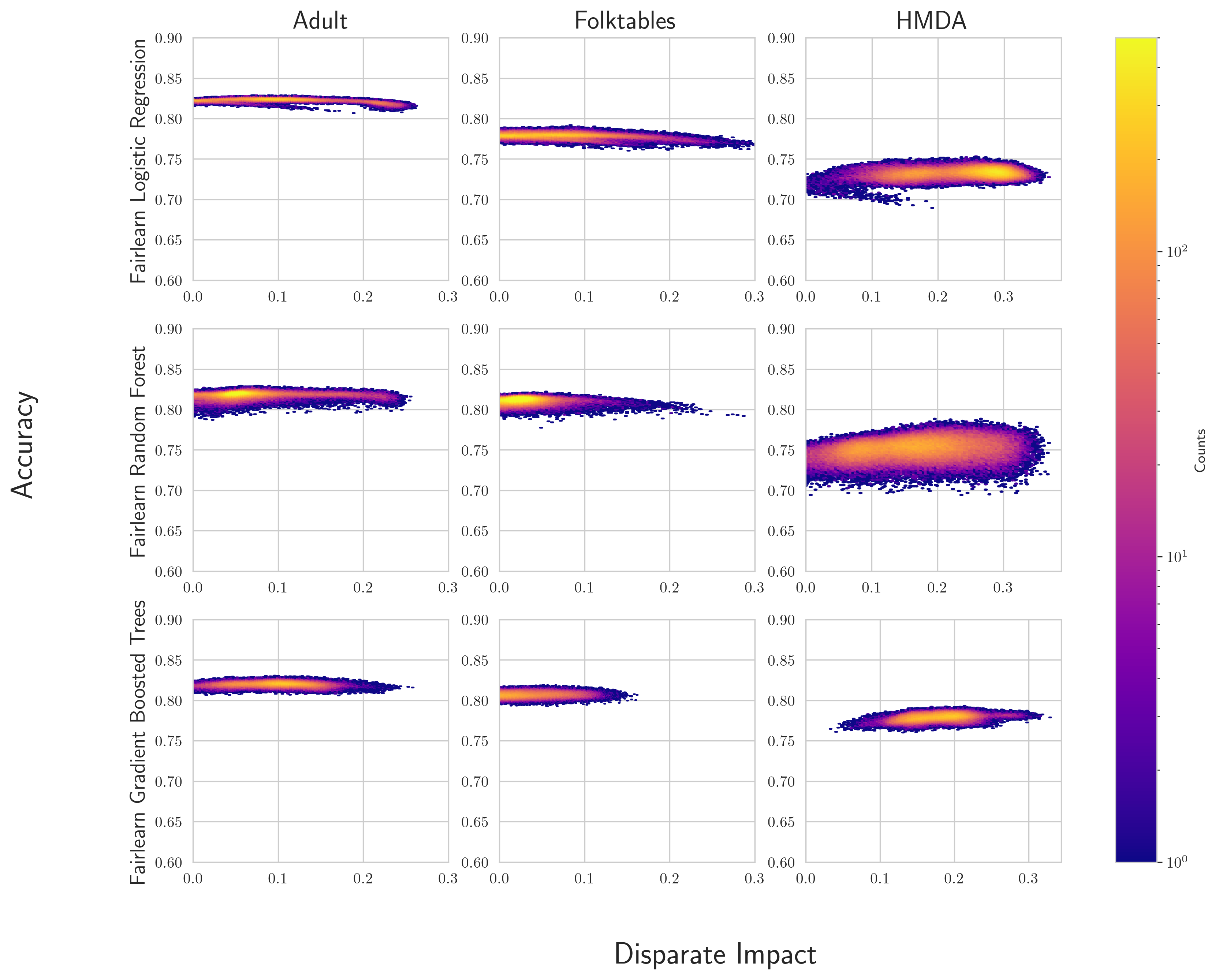}
    \caption{Version of \Cref{fig:acc_vs_srg} using the fairlearn methods. The disparate impact distribution is shifted to the left as a result of the fairness regularization. Like in \Cref{fig:acc_vs_srg}, there is substantial variation in disparate impact over the model re-training distribution.}
    \label{fig:fl_acc_vs_srg}
    \vspace{-0.1cm}
\end{figure}

\begin{figure}
    \centering
    \includegraphics[width=\linewidth]{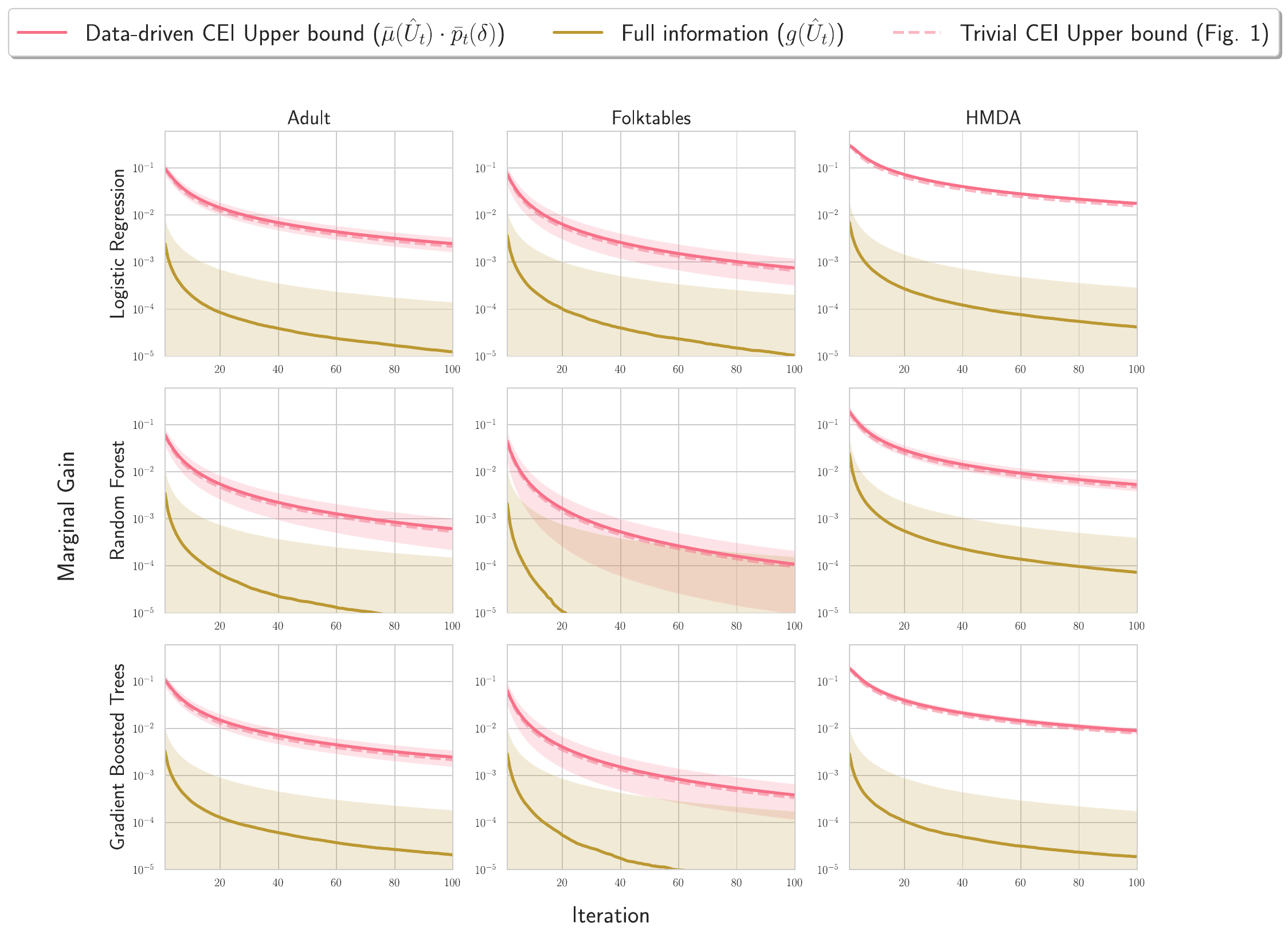}
    \caption{\Cref{alg:av3} run using the same setup as for \Cref{fig:cei}. For comparison, the upper bound from \Cref{fig:cei} is shown as a dashed line. }
    \label{fig:cei_mrl}
    \vspace{-0.1cm}
\end{figure}

\begin{figure}
    \centering
    \includegraphics[width=\linewidth]{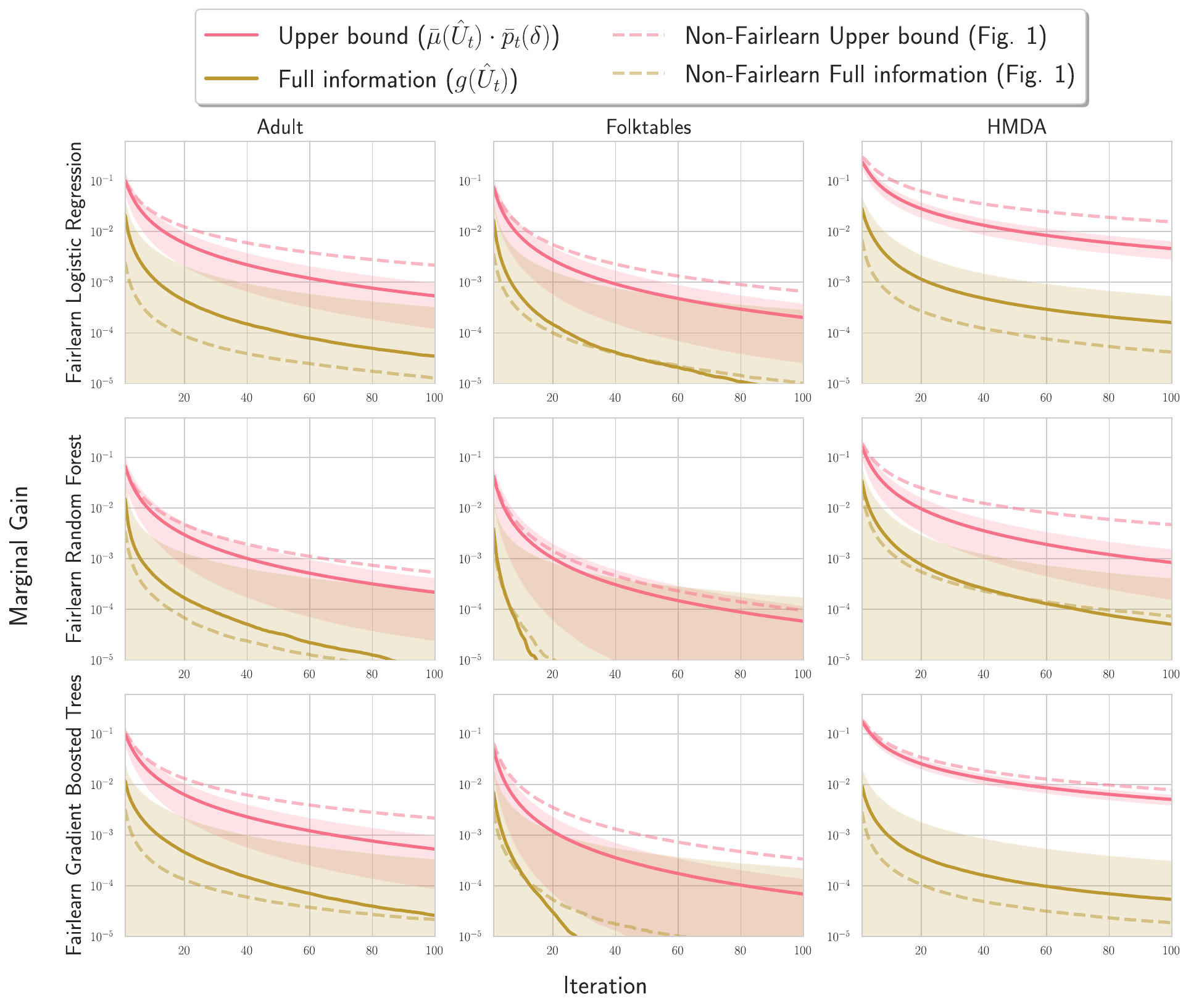}
    \caption{\Cref{alg:infinite} run on fairness-aware methods with the same datasets and base methods as in \Cref{fig:cei}. For comparison, the marginal gain trends from \Cref{fig:cei} are shown as dashed lines.}
    \label{fig:cei_fairlearn}
    \vspace{-0.1cm}
\end{figure}

\begin{figure}
    \centering
\begin{tabularx}{\linewidth}{p{7cm} X X X}
\toprule
Method & Adult & Folktables & HMDA \\
\midrule
Logistic Regression & 0.098 (0.010) & 0.073 (0.015) & 0.300 (0.021) \\
Random Forest & 0.060 (0.013) & 0.032 (0.014) & 0.192 (0.052) \\
Gradient Boosted Trees & 0.109 (0.013) & 0.063 (0.013) & 0.194 (0.013) \\
Fairlearn Logistic Regression & 0.100 (0.046) & 0.069 (0.042) & 0.241 (0.058) \\
Fairlearn Random Forest & 0.066 (0.037) & 0.029 (0.019) & 0.164 (0.067) \\
Fairlearn Gradient Boosted Trees & 0.096 (0.030) & 0.040 (0.025) & 0.177 (0.026) \\
\bottomrule
\end{tabularx}
    \caption{Selection rate disparities for each dataset and model class. Reported number is the mean over all models trained. Standard deviations are in parentheses.}
    \label{fig:srg}
\end{figure}

\begin{figure}
    \centering
\begin{tabularx}{\linewidth}{p{7cm} X X X}
\toprule
 & Adult & Folktables & HMDA \\
\midrule
Logistic Regression & 0.824 (0.044) & 0.780 (0.070) & 0.734 (0.140) \\
Random Forest & 0.818 (0.209) & 0.810 (0.228) & 0.748 (0.545) \\
Gradient Boosted Trees & 0.821 (0.088) & 0.808 (0.096) & 0.780 (0.101) \\
Fairlearn Logistic Regression & 0.824 (0.081) & 0.779 (0.096) & 0.733 (0.161) \\
Fairlearn Random Forest & 0.819 (0.148) & 0.812 (0.134) & 0.751 (0.453) \\
Fairlearn Gradient Boosted Trees & 0.821 (0.091) & 0.807 (0.101) & 0.780 (0.114) \\
\bottomrule
\end{tabularx}
    \caption{Accuracy for each dataset and model class. Reported number is the mean over all runs. Standard deviations are in parentheses.}
    \label{fig:acc}
\end{figure}

\paragraph{Empirical analysis of \Cref{alg:av3}.} We evaluate \Cref{alg:av3} using the same setup as for the analysis of \Cref{alg:infinite}.
An analogue to \Cref{fig:cei} for \Cref{alg:av3} is displayed in \Cref{fig:cei_mrl}.
There are not significant improvements from using the data-adaptive CEI bounds in \Cref{alg:av3} versus the trivial upper bounds in \Cref{alg:infinite}. In fact, the bounds get a small amount larger. This is a result of the fact that the data-adaptive CEI bounds require using some of the confidence budget ($\delta$) to estimate the CEI, which requires shrinking the part of the budget allocated to the improvement probability. This leads to slightly looser bounds on the improvement probability which counteract the perhaps tighter CEI bounds.
    
\paragraph{Empirical analysis of fairness-regularized methods.}
We display a figure analogous to \Cref{fig:cei} for the Fairlearn enabled methods in \Cref{fig:cei_fairlearn}.
We show the marginal gain trained from \Cref{fig:cei} for comparison.
Overall, the upper bounds with Fairlearn are less than with the vanilla methods, while the full information marginal gain is sometimes greater and sometimes lesser on average across our experiments.
The upper bounds may be lower in part because the overall disparate impact is closer to zero.
If the Fairlearn target disparate impact is set to a smaller value, like 0.05, we found in experiments that the marginal gains from retraining are smaller, leading to both the red and brown curves in the marginal gain plots to shift downward.

\paragraph{Miscoverage analysis.}
We next explore how well the coverage guarantees provided by the algorithm hold in practice.
A miscoverage event occurs when our anytime-valid upper bounds are violated: when the expected marginal gain from retraining is greater than the upper bound at some iteration during model retraining.
(We only trained for 100 iterations, so our estimates of miscoverage will be underestimates in circumstances when there would have been a violation after the 100th iteration.)
First, we evaluated miscoverage of the infinite data version of \cref{alg:infinite}, using the same setup as described above.
We detected less than 10 miscoverage events across all model retraining trajectories.
(There were 45 subsampled datasets and 5000 bootstrapped trajectories from our finite dataset of disparate impact performances, so there were $45 \times 5000 = 225000$ trajectories generated for each dataset and method combination.)

Next, we evaluate the finite data version of \Cref{alg:infinite}, where $Q_t$ is not observed exactly.
We find miscoverage significantly above the target 0.05 rate.
Miscoverage rates are displayed in \Cref{fig:miscoverage}.
The rates are particularly high for the Fairlearn methods and for methods trained on the Folktables data.

The explanation for high miscoverage is violations of \Cref{assm:regression} in our semi-synthetic evaluations.
We believe these assumptions violations are artifacts of our semi-synthetic evaluation setup, and that the assumption is likely to hold in typical circumstances when models are sampled from a distribution with infinite support.
In particular, since our population disparate impact distributions are computed as a discrete distribution (by sub-sampling data and training a fixed number of models), the selection effect is highly non-monotonic in some cases.
We show this in \Cref{fig:selectioneffect} for the vanilla methods and \Cref{fig:fl_selectioneffect} for the Fairlearn methods.
On the horizontal axis, we have percentile of $\hat Q_t$, where we sweep over the discrete distribution $\cD$.
On the vertical axis, we have the difference in selection effects, computed as
\begin{align}
    \E_{\mathbb{P}_0}[U_t - \hat U_t - U_{t+1} + \hat U_{t+1} \; | \; \hat U_t]. \label{eq:sectioneffect}
\end{align}
When this expression is greater than 0, for a particular $\hat U_t$, the non-decreasing selection effect holds.
We see that there are many violations of the assumption, and that violations are more frequent among model class and dataset combinations that have higher miscoverage, like the vanilla methods trained on Folktables, especially random forests, and the Fairlearn methods trained on Adult and Folktables.
Violations around smaller percentiles of $\hat Q_t$ are more problematic, since these are more likely to be selected as the minimum at a given iteration $t$, for large enough $t$.
Exploring how to gracefully account for violations of the assumption would be a valuable direction for future work.

In real-world data and model training distributions, we believe \Cref{assm:regression} is likely to hold.
In particular, in our semi-synthetic evaluation, $U_t$ conditional on $\hat U_t$ is typically not a random variable: Each $\hat U_t$ identifies its corresponding $U_t$ by virtue of the fact that a particular value of $\hat U_t$ is often unique in the dataset.
Thus, noise in realizations of $U_t \; | \; \hat U_t$ leads to non-smoothness in \Cref{eq:sectioneffect}.
In real-world settings, $U_t$ should be a true random variable, leading to smoothed versions of \Cref{fig:selectioneffect,fig:fl_selectioneffect} where all densities are shifted towards the higher density regions of the selection effect difference. 
(I.e., the purple regions of \Cref{fig:selectioneffect,fig:fl_selectioneffect} should shift towards the red and orange center.)
Since the red and orange high density regions of each plot are above zero, we'd expect true (non-semisynthetic) evaluations of the ground truth to adhere to \Cref{assm:regression}.
Similarly, there is idiosyncratic non-smoothness in $U_{t+1}$ and $\hat U_{t+1}$ in our semi-synthetic setup that likely wouldn't occur in real applications where the pool of possible models is much larger and where the disparate impact distribution should be much smoother.

\begin{figure}
    \centering
\begin{tabularx}{\linewidth}{p{7cm} X X X}
\toprule
 & Adult & Folktables & HMDA \\
\midrule
Logistic Regression & 0.006 & 0.161 & 0.000 \\
Random Forest & 0.077 & 0.381 & 0.000 \\
Gradient Boosted Trees & 0.006 & 0.196 & 0.000 \\
Fairlearn Logistic Regression & 0.219 & 0.400 & 0.000 \\
Fairlearn Random Forest & 0.273 & 0.484 & 0.045 \\
Fairlearn Gradient Boosted Trees & 0.146 & 0.475 & 0.000 \\
\bottomrule
\end{tabularx}
    \caption{Miscoverage for each dataset and model class. Reported number is the average number of model retraining trajectories where the true expected marginal gain was above the upper bound at some iteration.}
    \label{fig:miscoverage}
\end{figure}

\begin{figure}
    \centering
    \includegraphics[width=\linewidth]{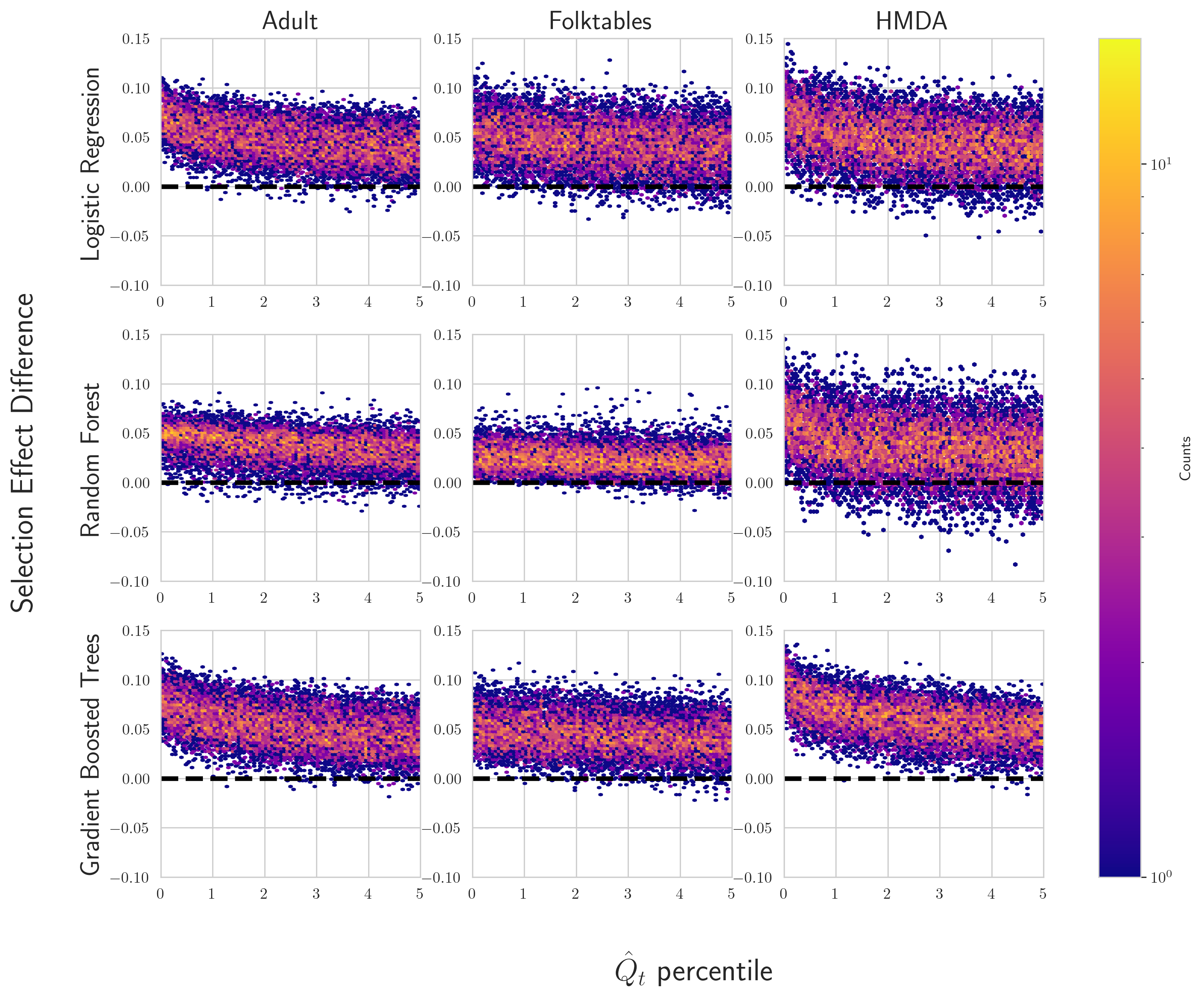}
    \caption{Selection effect difference ($\E_{\mathbb{P}_0}[U_t - \hat U_t - U_{t+1} + \hat U_{t+1} \; | \; \hat U_t]$) versus the percentile of $\hat Q_t$.}
    \label{fig:selectioneffect}
\end{figure}

\begin{figure}
    \centering
    \includegraphics[width=\linewidth]{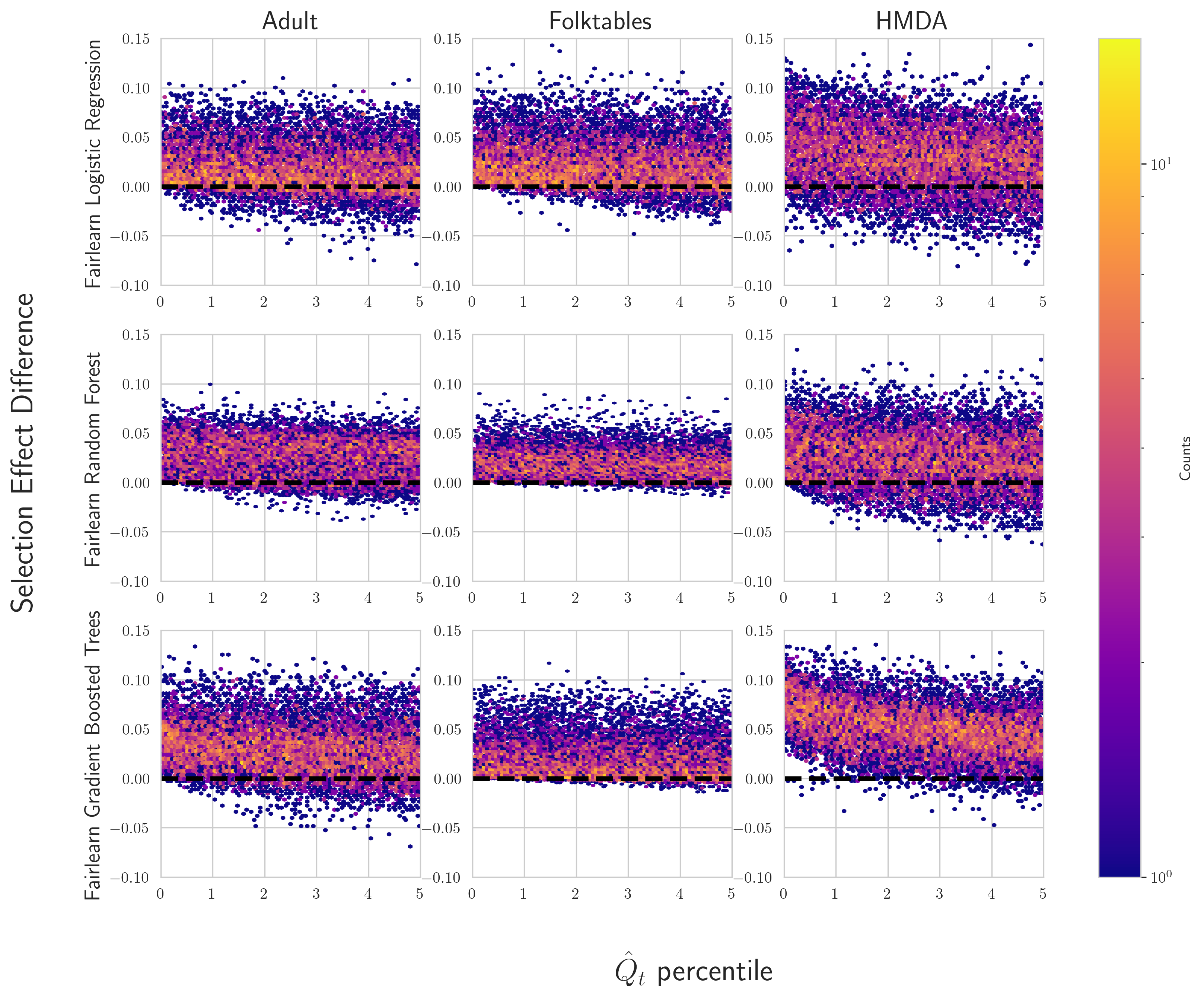}
    \caption{Version of \Cref{fig:selectioneffect} with Fairlearn model classes.}
    \label{fig:fl_selectioneffect}
\end{figure}

\section{Deferred proofs.} \label{sec:proofs}

Results are restated before proofs for reference.
    
\subsection{Deferred proofs for \Cref{sub:infinite}}

    \avmin*

\begin{proof}[Proof of \Cref{thm:always-valid-min}]
    At a high level, we proceed as follows:
    \begin{enumerate}
        \item First, show that it suffices to consider the case where the $X_t$ are uniform on $[0, 1]$ via the probability integral transform.
        \item Then, we show that it suffices to provide an anytime-valid upper bound on the running minimum of the sequence.
        \item Finally, we show that $\bar p_t$ as defined above yields such a bound.
    \end{enumerate}
    We begin by using the probability integral transform to ``convert'' our $X_t$'s into uniform random variables.
    Let $F_X$ be the CDF of $X_t$. Define
    \begin{align*}
        F_X^{-1}(u)
         = \inf\{x : F_X(x) \ge u\}.
    \end{align*}
    Let $\{U_t\}_{t=1}^\infty$ be iid uniform random variables on $[0, 1]$, defined on $\cP$.
    Let $V_t = \min_{s \in [t]} U_s$.
    Then, it holds, by e.g. Ch. 6, Theorem 3.1 of \citet{shorack2000probability}, that $$\{F_X^{-1}(U_t), F_X^{-1}(V_t)\}_{t=1}^\infty \stackrel{d}{=} \{X_t, Y_t\}_{t=1}^{\infty}.$$
    
    Because $F_X$ is monotone, $F_X^{-1}$ is monotone as well.
    Therefore,
    \begin{align*}
        &\cP_0[X_{t+1} < Y_t \given \{X_s\}_{s=1}^t] > \bar p_t(\alpha)
        \\&=
        \cP_0[X_{t+1} < Y_t \given Y_t] > \bar p_t(\alpha)
        \\
        &=
        \cP_0[F_X^{-1}(U_{t+1}) < F_X^{-1}(V_t) \given F_X^{-1}(V_t)] > \bar p_t(\alpha)
        \tag{$\{F_X^{-1}(U_t), F_X^{-1}(V_t)\}_{t=1}^\infty \stackrel{d}{=} \{X_t, Y_t\}_{t=1}^{\infty}$} \\
        &=
        \cP_0[F_X^{-1}(U_{t+1}) < F_X^{-1}(V_t) \given V_t] > \bar p_t(\alpha)
        \tag{$U_{t+1} \bot \{V_t\}$; $F_X^{-1}(V_t)$ is measurable with respect to $\sigma(F_X^{-1}(V_t)) \subseteq \sigma(V_t)$} \\
        &\le
        \cP_0[U_{t+1} < V_t \given V_t] > \bar p_t(\alpha)
        \tag{$F_X^{-1}$ is weakly increasing} \\
        &=
        V_t > \bar p_t(\alpha).
    \end{align*}
    The last inequality follows from the fact that $U_{t+1}$ is uniformly distributed on $[0, 1]$, so the probability it falls below $V_t$ is precisely $V_t$.
    Thus,
    \begin{align*}
        &\cP( \exists t \in \N \; : \; \cP_0(X_{t+1} < Y_t \given \{X_s\}_{s=1}^t) \\
        > &\bar p_t(\alpha)) \leq \cP( \exists t \in \N \; : \; V_t > \bar p_t(\alpha))
    \end{align*}
    Thus, our goal is now to provide an anytime-valid upper-bound on $V_t$.
    
    Define the martingale
    \begin{align}
        M_t(\theta) &\defeq \frac{1}{(1 - \theta)^t} \indic{V_t \geq \theta} \nonumber \\
        &= M_{t-1}(\theta) \cdot \paren{\frac{1}{1-\theta} \indic{U_t \ge \theta}}.
    \end{align}
    This is a martingale because
    \begin{align*}
        &\E[M_t \given M_{t-1}] \\
        &= \Eb{M_{t-1}(\theta) \paren{\frac{1}{1-\theta} \indic{U_t \ge \theta}} \; \bigg| \; M_{t-1}(\theta)} \\
        &= \frac{M_{t-1}(\theta)}{1-\theta} \Eb{\paren{\indic{U_t \ge \theta}} \; \bigg| \; M_{t-1}(\theta)} \\
        &= \frac{M_{t-1}(\theta)}{1-\theta} \Pr[U_t \ge \theta] \\
        &= M_{t-1}(\theta).
    \end{align*}
    Moreover, it is a test martingale because it is nonnegative.
    Next, we use the ``method of mixtures'' \citep[see, e.g.,~][]{robbins1970statistical,waudby2024estimating} to mix $M_t$ with a uniform distribution on $\theta$ over $[0, 1]$.
    Intuitively, placing more mass on smaller values of $\theta$ gives us sharper bounds for larger values of $t$.
    We choose the uniform distribution here for simplicity.
    In \Cref{thm:always-valid-min-iterated-log}, we show how to get an asymptotically tight rate.
    \begin{align*}
        M_t^U(\theta)
        &\triangleq \int_0^1 M_t(\theta) \, d\theta \\
        &= \int_0^1 \frac{1}{(1-\theta)^t} \indic{V_t \ge \theta} \, d\theta \\
        &= \int_0^{V_t} \frac{1}{(1-\theta)^t} \, d\theta.
        \numberthis \label{eq:mix-def}
    \end{align*}
    By Fubini's theorem, this is also a test martingale.
    Applying Ville's inequality (\Cref{thm:ville}), for any $\alpha \in (0, 1)$,
    \begin{align*}
        \cP \paren{\exists t \in \N \; : \; ~ M_t^U(\theta) > \frac{1}{\alpha}}
        &\le \alpha\\
        \cP\paren{\exists t \in \N \; : \; \int_0^{V_t} \frac{1}{(1-\theta)^t} > \frac{1}{\alpha} }
        &\le \alpha.
        \numberthis \label{eq:ville-applied}
    \end{align*}
    Observe that the integrand is nonnegative, so for any sequence $\bar p_t$,
    \begin{equation}
        \int_0^{V_t} \frac{1}{(1-\theta)^t} \, d\theta
        >
        \int_0^{\bar p_t} \frac{1}{(1-\theta)^t} \, d\theta
        \Longleftrightarrow V_t > \bar p_t.
        \label{eq:Yt-pt-iff}
    \end{equation}
    Therefore, we can choose $\bar p_t(\alpha)$ such that
    \begin{equation}
        \int_0^{\bar p_t(\alpha)} \frac{1}{(1-\theta)^t} \, d\theta
        = \frac{1}{\alpha}.
        \label{eq:pt-condition}
    \end{equation}
    For $t = 1$,
    \begin{align*}
        \int_0^{\bar p_1(\alpha)} \frac{1}{1-\theta} \, d\theta
        &= \frac{1}{\alpha} \\
        -\log(1 - \bar p_1(\alpha))
        &= \frac{1}{\alpha} \\
        \bar p_1(\alpha)
        &= 1 - e^{-1/\alpha}.
    \end{align*}
    For $t \ge 2$,
    \begin{align*}
        \int_0^{\bar p_t(\alpha)} \frac{1}{(1-\theta)^t} \, d\theta
        &= \frac{1}{\alpha} \\
        \frac{(1 - \bar p_t(\alpha))^{t-1} - 1}{t-1}
        &= \frac{1}{\alpha} \\
        \bar p_t(\alpha)
        &= 1 - \paren{\frac{t-1}{\alpha} +1}^{-1/(t-1)}
    \end{align*}
    Thus,
    \begin{align*}
        &\cP\paren{\exists t \in \N \; : \; ~ V_t > \bar p_t(\alpha) } \\
        &= \cP\paren{
        \exists t \in \N :  \int_0^{V_t} \frac{1}{(1-\theta)^t} d\theta
        >
        \int_0^{\bar p_t} \frac{1}{(1-\theta)^t}  d\theta
        }
        \tag{by~\cref{eq:Yt-pt-iff}}
        \\
        &= \cP\paren{
        \exists t\in \N \; : \;  ~ \int_0^{V_t} \frac{1}{(1-\theta)^t} \, d\theta
        >
        \frac{1}{\alpha} 
        }
        \tag{by~\cref{eq:pt-condition}}
        \\
        &\le \alpha,
        \tag{by~\cref{eq:ville-applied}}
    \end{align*}
    completing the proof.
\end{proof}

    We first prove a theorem for general iid random variables bounded in $[0, 1]$.
    \begin{theorem}\label{thm:generic}
        For all $\gamma, \delta > 0$ and $\cP$, \Cref{alg:infinite} run with $\cP$, $\gamma,\delta$ and any $ \bar\mu$ satisfying \Cref{def:bounded-mrl} as input terminates at a stopping time $\tau \in \N$ such that
        \begin{align*}
            \cP( \E_{\cP}[X_\tau - X_{\tau + 1} \given X_\tau ] < \gamma ) \geq 1 - \delta.
        \end{align*}
    \end{theorem}

    \begin{proof}[Proof of \Cref{thm:generic}]
        Observe:
        \begin{align*}
            &\cP( \E[X_\tau - X_{\tau + 1} \given U_\tau] > \gamma) \\
            &=
            \cP(g(X_\tau) > \gamma) \\
            &= \cP(\mu(X_\tau) p(X_\tau) > \gamma) \\
            &\le \cP(\mu(X_\tau) p(X_\tau) > \bar \mu(X_\tau) \bar p_\tau(\delta) ) \tag{$\bar \mu(X_\tau) \bar p_\tau(\delta) \le \gamma$ by the stopping condition} \\
            &\le \cP(\bar \mu(X_\tau) p(X_\tau) > \bar \mu(X_\tau) \bar p_\tau(\delta) ) \tag{$\mu(X_\tau) \le \bar \mu(X_\tau)$ almost surely} \\
            &= \cP(p(X_\tau) > \bar p_\tau(\delta) ) \tag{$\bar \mu (u) \geq 0$ for all $u$}  \\
            &\le \cP(\exists t \in \N \; : \; p(X_t) > \bar p_t(\delta) ) \\
            &\le \delta \tag{\Cref{thm:always-valid-min}}
        \end{align*}
    \end{proof}

    \begin{theorem}[Ville's inequality] \label{thm:ville}
        Let $M_1,M_2,\dots$ be a non-negative supermartingale scaled so that $\E M_1 \leq 1$. Then, for any real number $\alpha$, 
        \begin{align*}
            P \paren{\sup_{t \geq 1} M_t \geq \frac{1}{\alpha}} \leq \alpha.
        \end{align*}
    \end{theorem}

\subsection{Deferred proofs of \Cref{sub:finite}}

\finite*

\begin{proof}[Proof of \Cref{thm:finite}.]
    First, observe:
    \begin{align*}
        &\E_{\bbP_0} [U_{\tau} - U_{\tau + 1} \given \hat U_\tau] \\ 
        &= \E_{\bbP_0} [\hat U_{\tau} - \hat U_{\tau + 1} \given \hat U_\tau] \\
        &\quad+ \E_{\bbP_0} [(U_{\tau} - \hat U_\tau) - (U_{\tau + 1} - \hat U_{\tau + 1}) \given \hat U_\tau].
    \end{align*}
    Under \cref{assm:regression}, 
    \begin{align}
        \E_{\bbP_0} [(U_{\tau} - \hat U_\tau) - (U_{\tau + 1} - \hat U_{\tau + 1}) \given \hat U_\tau] \geq 0, \label{eq31}
    \end{align}
    for all $t$ with probability 1. 
    Next, observe 
    \begin{align}
        \E_{\bbP_0} [\hat U_{\tau} - \hat U_{\tau + 1} \given \hat U_\tau] &= \E_{\hat P_0} [\hat U_{\tau} - \hat U_{\tau + 1} \given \hat U_\tau]. \label{eq32}
    \end{align}
    Finally, from \Cref{thm:generic}, we have
    \begin{align}
        \hat P(\E_{\hat P_0} [\hat U_{\tau} - \hat U_{\tau + 1} \given \hat U_\tau] \leq \gamma) \geq 1 - \delta. \label{eq33}
    \end{align}
    Putting it all together, we have 
    \begin{align*}
         &\bbP(\E_{\bbP_0} [U_{\tau} - U_{\tau + 1} \given \hat U_\tau] \leq \gamma) \\ &\geq \bbP(\E_{\bbP_0} [\hat U_{\tau} - \hat U_{\tau + 1} \given \hat U_\tau] \leq \gamma) \tag{\Cref{eq31}} \\
         &= \hat P(\E_{\hat P_0} [\hat U_{\tau} - \hat U_{\tau + 1} \given \hat U_\tau] \leq \gamma) \tag{\Cref{eq32}} \\
         &\geq 1-\delta.\tag{\Cref{eq33}}
    \end{align*}
\end{proof}

\subsection{Deferred proofs for \Cref{sub:est-mrl}}

\estmrl*

\begin{proof}[Proof of \Cref{thm:est-mrl}]
    Define the following events.
    \begin{align*}
        \cE_0 &=\{ C \leq \textnormal{\median}(\hat P_0) \} \\
        \cE_1 &= \{ \E_{\hat P_0} [C - \hat Q_{\tau + 1} \; | \; C > \hat Q_{\tau + 1}, C] \leq \bar \mu_\tau \}  
    \end{align*}
    where $\mu_\tau$ is as defined in \cref{alg:av3}.
    
    Notice that, on $\cE_0$ and $\cE_1$ 
    \begin{align*}
        &\E_{\hat P_0} [\hat U_{\tau} - \hat U_{\tau + 1} \; | \; \hat U_{\tau} > \hat Q_{\tau + 1}]  \\
        &\leq \E_{\hat P_0} [z - \hat Q_{\tau + 1} \; | \; z > \hat Q_{\tau + 1}]  \tag{\cref{assm:mrl}} \\
        \implies &\E_{ \hat P_0} [\hat U_{\tau} - \hat Q_{\tau + 1} \; | \; \hat U_{\tau} > \hat Q_{\tau + 1}] \\
        &\leq \E_{\hat P_0} [C - \hat Q_{\tau + 1} \; | \; C > \hat Q_{\tau + 1}, C] \tag{$C \in [\hat U_\tau, \median(P_0)]$ a.s. on $\cE_0$} \\
        &\leq \bar \mu_\tau. \tag{$\cE_1$}
    \end{align*}
    where $\bar \mu_\tau$ is defined as in \Cref{alg:av3}.
    Also, define $\cE_2 = \{ \hat P_0(\hat U_{\tau} > \hat U_{\tau+1} \given \hat U_\tau ) \leq \bar p_\tau \} $.
    Observe that, on $\cE_2$,
    \begin{align*}
        \hat P_0(\hat U_{\tau} > \hat Q_{\tau+1} \given \hat U_\tau) \leq \bar p_\tau(\delta / 3).
    \end{align*}
    Combining these, we have, by the fact that the algorithm terminated
    \begin{align}
        \bar \mu_\tau \cdot \bar p_\tau(\delta / 3) \leq \gamma.
    \end{align}
    By \Cref{lem:medianestimation,lem:mrl,thm:always-valid-min}, $\cE_0, \cE_1$ and $\cE_2$ each occur with probability at least $1-\delta/3$, so by a union bound, their intersection occurs with probability at least $1-\delta$.
\end{proof}

\begin{lemma} \label{lem:medianestimation}
    For all $\delta$, with probability no less than $1-\delta/3$, 
    \begin{align}
        C \leq \textnormal{\median}(P_0)
    \end{align}
    where $C$ is defined as in \Cref{alg:av3}.
\end{lemma}
\begin{proof}[Proof of \Cref{lem:medianestimation}]
    Let $\varepsilon = 1/6$ and let $i^* = \floor{T_1/2}$.
    Note that the event $C \leq \median(\hat P_0)$ is the same as the event that $i^* \leq \sum_{t=1}^{T_1} \indic{\hat Q_t \leq \median( \hat P_0)}$, since this implies that there are at least $i^*$ draws of $\hat Q_t$ less than the median.
    Note that $\indic{\hat Q_t \leq \median(\hat P_0)}$ are independent and distributed as Bernoulli random variables with success probability $p$.
    Thus, 
    \begin{align*}
        \hat P(C > \median(\hat P_0)) &= \hat P \paren{i^* > \sum_{t=1}^{T_1} \indic{\hat Q_t \leq \median(\hat P_0)}} \\
        &\leq \exp \paren{- \frac{2 (i^* - (1/2-\varepsilon) T_1)^2}{T_1}} \tag{Hoeffding's inequality} \\
        &\leq \exp \paren{- {2 \varepsilon^2 T_1}}  \tag{Substituting definition of $i^*$.} \\
        &\leq \frac{\delta}{3}  \tag{Substituting definition of $\varepsilon$ and simplifying.}
    \end{align*}
\end{proof}

\begin{lemma} \label{lem:mrl}
    For all $\delta$, with probability at least $1 - \delta/3$, it holds for all $t = 2,3, \dots$ simultaneously that
    \begin{align*}
        \E_{\hat P} [C - \hat Q_{t + 1} \; | \; C > \hat Q_{t + 1}, C] \leq \bar \mu_t
    \end{align*}
    where $\bar \mu_t$ is defined as in \Cref{alg:av3}.
\end{lemma}

\begin{proof}[Proof of \Cref{lem:mrl}]
    We just need to verify that we can apply \Cref{cor:mueb}.
    To do this, we need to verify $\Delta \in S_t$ have the same conditional mean.

    Define $S$ and $\{ i_t \}_t$ analogously to in \Cref{cor:mueb}:
    \begin{align*}
        S = \{ t \in \N \; : \; \hat Q_t < C\}.
    \end{align*}

    Define the sequence $S = (t \in \N \; : \; \hat Q_t < C)$.
    To see that all $\Delta_{i_t}$ have the same mean conditional on the past, observe, 
    \begin{align*}
        \E_{\hat P}[\Delta_{i_t} \; | \; \Delta_{i_1}, \dots, \Delta_{i_{t-1}}, C] &= \E_{\hat P}[C - \hat Q_{i_t} \; | \; \Delta_{i_1}, \dots, \Delta_{i_{t-1}}, C]\\
        &=C - \E_{ \hat P}[\hat Q_{i_t} \; | \; \Delta_{i_1}, \dots, \Delta_{i_{t-1}}, C]\\
        &=C - \E_{\hat P}[\hat Q_{i_t}]  \tag{{Independence of $\hat Q_{i_t}$ conditional on $D$}}
    \end{align*}
    Thus, since $\hat Q_{i_t}$ are identically distributed conditional on $D$, it holds $\E_{\hat P_0} \hat Q_{i_t} = \E_{\hat P_0} \hat Q_{i_s}$ for all $s, t \in \N$ so $\{ \Delta_{i_t} \}_{t=1}^\infty$ have the same mean conditional on the past and $C$.

    Now, on the event that $\hat Q_{t+1} < C$, it holds $t+1 \in S$.
    Thus, the guarantee holds for $t+1$.
    Finally, we plug in $\delta/3$ for $\alpha$, which yields the desired result:
    \begin{align*}
        \E_{\hat P} [C - \hat Q_{t + 1} \; | \; C > \hat Q_{t + 1}, C] \leq \bar \mu_t.
    \end{align*}
\end{proof}

    The following result provides a high probability upper bound for anytime-valid bounded mean estimation.
    \begin{theorem}[Theorem 2, \citet{waudby2024estimating}] \label{thm:ramdas}
        Suppose there is a constant $\nu$ and stochastic process $(X_t)_{t=1}^\infty \sim \cP$ for some distribution $\cP$ with support bounded on $[0, 1]$ such that, for all $t$,
        \begin{align*}
            \E_\cP(X_t \; | \; X_1, \dots, X_{t-1}) = \nu.
        \end{align*}
        Let $\cF_{t} = \sigma(\{ X_i \}_{i=1}^t)$ be the $\sigma$-field induced by $X_1, \dots, X_t$.
        Next, consider any sequence $\{ \lambda_t \}_{t=1}^\infty$ such that for all $t$, $\lambda_t$ is $\cF_{t-1}$-measurable. Then, for all $\alpha > 0$, with probability at least $1-\alpha$, it holds for all $t=1,2,\dots$ simultaneously:
        \begin{align*}
            \nu \leq \frac{\log(2/\alpha) + \sum_{i=1}^t \lambda_i X_i - (X_i - \hat \nu_{i-1})^2 (\log(1 - \lambda_i) + \lambda_i)}{\sum_{i=1}^t \lambda_i}
        \end{align*}
\end{theorem}

We state the following corollary \Cref{thm:ramdas} which states the result for subsequences of random processes (which amounts to a re-indexing) and uses a particular choice of $\lambda_t$.
This result follows the recommendations for $\lambda_t$ in \citet{waudby2024estimating} and is an empirical Bernstein-type bound.
\begin{corollary} \label{cor:mueb}
    Suppose there is a constant $\nu$ and stochastic process $(X_t)_{t=1}^\infty \sim \cP$ for some distribution $\cP$ with support bounded on $[0, 1]$. 
    Define a sequence of subsets $S_t$ such that $S_{t-1} \subseteq S_t$ and $S_t \setminus S_{t-1} \subseteq \{ t \}$.
    Suppose, for all $t$ such that $t \in S_t$ and $i \in S_t$,
        \begin{align*}
            \E_\cP(X_t \; | \; S_{t-1}) = \nu.
        \end{align*}
    For all $\alpha \in (0, 1]$, define 
    \begin{align}
        &{\lambda_t
        \defeq \min\left\{\sqrt{\frac{2 \log (2/\alpha)}{\hat \sigma_{t-1}^2 \abs{S_t} \log(1+\abs{S_t})}}, \frac{1}{2} \right\}} \label{def:lambda}
    \end{align}
    where 
    \begin{align*}
        &{\hat \nu_t
        \defeq \frac{\frac{1}{2} + \sum_{i \in S_t} X_i}{1 + \abs{S_t}}}, \text{ and} \\
        &{\hat \sigma_t^2
        \defeq \frac{\frac{1}{4} + \sum_{i \in S_t} (X_i  - \hat \nu_i)^2}{1 + \abs{S_t}}}.
    \end{align*}
    Finally, for all $t$, let 
    \begin{align}
        \mueb_t(\{ X_s \}_{s=1}^t, \alpha, S_t) \defeq \frac{\log(2/\alpha) + \sum_{i \in S_t}\lambda_i X_i - (X_i - \hat \nu_{i-1})^2 (\log(1 - \lambda_i) + \lambda_i)}{\sum_{i \in S_t} \lambda_i}. \label{def:mueb}
    \end{align}
    Then, with probability at least $1 -\alpha$, it holds for all $t=1,2,\dots $ simultaneously:
    \begin{align*}
        \nu \leq \mueb_t(\{ X_s \}_{s=1}^t, \alpha, S_t)
    \end{align*}
\end{corollary}

\begin{proof}[Proof of \Cref{cor:mueb}]
    Define the sequence $S = (t \in \N \; : \; X_t \in S_t)$.
    Denote by $i_t$ the $t$-th element of $S$.
    Clearly, $\lambda_t$ is $\cF_{t-1}$-measurable.
    To apply the theorem, we plug in the sequence $\{ X_{i_s} \}_{s=1}^{\abs{S_t}}$ as defined in for $X_t$ in \Cref{thm:ramdas}.
\end{proof}

\section{A sharper upper bound for \Cref{thm:always-valid-min} with an almost matching lower bound.}

\begin{theorem}
    Let $\{U_t\}_{t=1}^\infty$ be a sequence of iid uniform random variables on $[0, 1]$.
    Let $V_t = \min_{s \in [t]} U_s$.
    For any constant $\varepsilon > 0$,
    define\footnote{By convention, $\inf \varnothing = \infty$.}
    \begin{align*}
        \tilde p_t(\delta) \triangleq \min\left\{1, \inf_{q \in [0, e^{-1})}\left\{\int_0^{q} \frac{1}{(1-\theta)^t}
        \frac{\varepsilon}{\theta \cdot (\log(1/\theta))^{1+\varepsilon}}
        \, d\theta
        \ge \frac{1}{\delta}
        \right\}
        \right\}.
    \end{align*}
    Then,
    \begin{align*}
        \Pr\{\exists t ~ V_t > \tilde p_t(\delta)\} \le \frac{1}{\delta}.
        \numberthis \label{eq:av-ptilde}
    \end{align*}
    Asymptotically,
    \begin{align*}
        \lim_{t \to \infty} \frac{\tilde p_t(\delta)}{\frac{\log \log t}{t}} \in [1, 1+\varepsilon].
    \end{align*}
    Moreover, this is nearly tight: for any sequence $\{q_t\}_{t=1}^\infty$,
    \begin{align*}
        \Pr\{\exists t ~ V_t > q_t\} \le \frac{1}{\delta}
        \Longrightarrow
        \lim_{t \to \infty} \frac{q_t}{\frac{\log \log t}{t}} \ge 1.
    \end{align*}
    \label{thm:always-valid-min-iterated-log}
\end{theorem}
\begin{proof}
    The lower bound follows directly from \citet[][Theorem 1]{robbins1972law}, which states that
    \begin{align*}
        \Pr\left\{
        V_t \ge \frac{\log \log t + 2 \log \log \log t}{t} ~ ~ ~ \text{i.o.}
        \right\}
        =1.
    \end{align*}
    For any $\{q_t\}_{t=1}^\infty$ such that
    \begin{align*}
        \lim_{t \to \infty} \frac{q_t}{\frac{\log \log t}{t}} < 1,
    \end{align*}
    there is some $t^*$ such that for all $t \ge t^*$, $q_t < \frac{\log \log t + 2 \log \log \log t}{t}$.
    But this means that $V_t > q_t$ infinitely often for $t \ge t^*$, so $\Pr\{\exists t ~ V_t > q_t\} = 1$.

    For our upper bound, we follow the proof of \Cref{thm:always-valid-min} to define the test martingale
    \begin{align*}
        M_t(\theta)
        &\triangleq \frac{1}{(1-\theta)^t} \indic{V_t \ge \theta}.
    \end{align*}
    In \Cref{thm:always-valid-min}, we mixed this martingale over the uniform distribution over $[0, 1]$ for $\theta$.
    This lead to an asymptotically loose bound:
    \begin{align*}
        \bar p_t(\delta)
        &= 1 - \paren{\frac{t-1}{\delta} + 1}^{-1/(t-1)} \\
        &= 1 - \exp\b{-\frac{1}{t-1} \log \paren{\frac{t-1}{\delta} + 1}}.
    \end{align*}
    By \Cref{lem:asymp-simplify},
    \begin{align*}
        \bar p_t(\delta)
        &\sim \frac{1}{t-1} \log \paren{\frac{t-1}{\delta} + 1} \\
        &\sim \frac{\log t}{t}.
    \end{align*}

    To get something asymptotically tight, we need to mix with a distribution that places more mass on very small values of $\theta$.
    For some constant $\varepsilon > 0$, consider the distribution
    \begin{align*}
        \nu(\theta)
        &\triangleq \frac{\varepsilon}{\theta \cdot (\log(1/\theta))^{1+\varepsilon}}
    \end{align*}
    defined on $(0, e^{-1})$.
    This is a valid probability distribution because
    \begin{align*}
        \int_0^{e^{-1}} \nu(\theta) \, d\theta
        &= \int_0^{e^{-1}} \frac{\varepsilon}{\theta \cdot (\log(1/\theta))^{1+\varepsilon}} \, d\theta \\
        &= \varepsilon \int_1^{\infty} u^{-1-\varepsilon} \, du
        \tag{substitute $u = \log(1/\theta)$} \\
        &= \varepsilon \cdot \frac{-1}{\varepsilon} u^{-\varepsilon} \bigg|_{1}^{\infty} \\
        &=  u^{-\varepsilon} \bigg|_{\infty}^{1} \\
        &=  1.
    \end{align*}

    We define our test martingale to be a mixture of $M_t$ over this distribution $\nu$:
    \begin{align*}
        M_t^N(\theta)
        &\triangleq \int_0^{e^{-1}} M_t(\theta) \nu(\theta) \, d\theta \\
        &= \varepsilon \int_0^{\min(e^{-1}, V_t)} \frac{1}{(1-\theta)^t} \frac{1}{\theta (\log(1/\theta))^{1+\varepsilon}} \, d\theta.
    \end{align*}
    Again, this is a nonnegative martingale by Fubini's theorem, using the fact that $M_t(\theta)$ is a nonnegative martingale as shown in the proof of \Cref{thm:always-valid-min}.
    Applying Ville's inequality (\Cref{thm:ville}), for any $\delta \in (0, 1)$,
    \begin{align*}
        \Pr\left\{
        \exists t ~ M_t^N(\theta) > \frac{1}{\delta}
        \right\}
        \le \delta.
    \end{align*}

    Define
    \begin{align*}
        \tilde p_t(\delta) \triangleq \min\left\{1, \inf\left\{q \in [0, e^{-1}) :
        \int_0^{q} \frac{1}{(1-\theta)^t}
        \frac{\varepsilon}{\theta \cdot (\log(1/\theta))^{1+\varepsilon}}
        \, d\theta
        \ge \frac{1}{\delta}
        \right\}
        \right\}.
    \end{align*}
    For sufficiently large $t$, the set over which we are taking the infimum will be nonempty, and for such $t$,
    \begin{align*}
        \int_0^{\tilde p_t(\delta)} \frac{1}{(1-\theta)^t}
        \frac{\varepsilon}{\theta \cdot (\log(1/\theta))^{1+\varepsilon}}
        \, d\theta
        = \frac{1}{\delta}.
    \end{align*}

    By a simple monotonicity argument,
    \begin{align*}
        \{V_t > \tilde p_t(\delta) \}
        &\Longleftrightarrow
        \{V_t > \tilde p_t(\delta), \tilde p_t(\delta) < e^{-1} \} \\
        &\Longrightarrow
        \varepsilon \int_0^{\min(e^{-1}, V_t)} \frac{1}{(1-\theta)^t} \frac{1}{\theta (\log(1/\theta))^{1+\varepsilon}} \, d\theta > \frac{1}{\delta} \\
        &\Longleftrightarrow \left\{M_t^N > \frac{1}{\delta}\right\}.
    \end{align*}
    Therefore,
    \begin{align*}
        \Pr\left\{
        \exists t ~ V_t > \tilde p_t(\delta)
        \right\}
        &\le
        \delta,
    \end{align*}
    which proves~\Cref{eq:av-ptilde}.

    Finally, we provide the required asymptotic equivalence.
    For sufficiently large $t$,
    \begin{align*}
        \frac{1}{\delta}
        &=
        \int_0^{\tilde p_t(\theta)}
         \frac{1}{(1-\theta)^t}
        \frac{\varepsilon}{\theta \cdot (\log(1/\theta))^{1+\varepsilon}} \, d\theta \\
        &= 
        \int_0^{1/t}
         \frac{1}{(1-\theta)^t}
        \frac{\varepsilon}{\theta \cdot (\log(1/\theta))^{1+\varepsilon}} \, d\theta
        +
        \int_{1/t}^{\tilde p_t(\theta)}
         \frac{1}{(1-\theta)^t}
        \frac{\varepsilon}{\theta \cdot (\log(1/\theta))^{1+\varepsilon}} \, d\theta \\
        &\ge 
        \int_{1/t}^{\tilde p_t(\theta)}
         \frac{1}{(1-\theta)^t}
        \frac{\varepsilon}{\theta \cdot (\log(1/\theta))^{1+\varepsilon}} \, d\theta.
        \tag{The integrand is nonnegative.}
    \end{align*}
    Next, we bound
    \begin{align*}
        \int_{1/t}^{\tilde p_t(\theta)}
         \frac{1}{(1-\theta)^t}
        \frac{\varepsilon}{\theta \cdot (\log(1/\theta))^{1+\varepsilon}} \, d\theta
        &= \varepsilon \int_1^{t \tilde p_t(\theta)} \frac{1}{(1 - v/t)^t} \frac{1}{v/t (\log (t/v))^{1+\varepsilon}} \, \frac{dv}{t}
        \tag{Substitute $v = \theta t$}
        \\
        &= \varepsilon \int_1^{t \tilde p_t(\theta)} \frac{1}{(1 - v/t)^t} \frac{1}{v (\log (t/v))^{1+\varepsilon}} \, dv \\
        &\ge
        \varepsilon \int_1^{t \tilde p_t(\theta)} e^{v} \frac{1}{v (\log (t/v))^{1+\varepsilon}} \, dv
        \tag{$1/(1-v/t)^t \ge e^{v}$ for $v \ge 1$}
        \\
        &=
        \frac{\varepsilon}{(\log t)^{1+\varepsilon}} \int_1^{t \tilde p_t(\theta)} e^{v} \frac{1}{v (1 - \frac{\log v}{\log t})^{1+\varepsilon}} \, dv
        \\
        &\ge
        \frac{\varepsilon}{(\log t)^{1+\varepsilon}} \int_1^{t \tilde p_t(\theta)} \frac{e^{v}}{v} \, dv.
        \tag{$\log v \ge 0$ for $v \ge 1$}
    \end{align*}
    Here, we have used the fact that $\tilde p_t(\theta) < 1$ for large $t$.

    Putting this together, we have
    \begin{align*}
        \frac{1}{\delta}
        &\ge \frac{\varepsilon}{(\log t)^{1+\varepsilon}} \int_1^{t \tilde p_t(\delta)} \frac{e^v}{v} \, dv.
    \end{align*}
    Consider the sequence $z_t$ implicitly defined as
    \begin{align*}
        \frac{\varepsilon}{(\log t)^{1+\varepsilon}} \int_1^{z_t} \frac{e^v}{v} \, dv
        &= \frac{1}{\delta}.
    \end{align*}
    Clearly, $z_t \ge t \tilde p_t(\delta)$ because the integrand $e^v/v$ is nonnegative.
    We will show that $z_t \sim (1+\varepsilon) \log \log t$.

    The exponential integral $\Ei$ is defined
    \begin{align*}
        \Ei(x)
        &\triangleq \int_{-\infty}^x \frac{e^v}{v} \, dv.
    \end{align*}
    Therefore,
    \begin{align*}
        \frac{\varepsilon}{(\log t)^{1+\varepsilon}} \int_1^{z_t} \frac{e^v}{v} \, dv
        &= \frac{\varepsilon}{(\log t)^{1+\varepsilon}} (\Ei(z_t) - \Ei(1)).
    \end{align*}
    By definition of $z_t$, for all $t$,
    \begin{align*}
        \frac{\delta \varepsilon}{(\log t)^{1+\varepsilon}} (\Ei(z_t) - \Ei(1))
        &= 1.
    \end{align*}
    Therefore,
    \begin{align*}
        \lim_{t \to \infty}
        \frac{\delta \varepsilon}{(\log t)^{1+\varepsilon}} (\Ei(z_t) - \Ei(1))
        &= 1 \\
        \lim_{t \to \infty}
        \frac{\delta \varepsilon}{(\log t)^{1+\varepsilon}} \Ei(z_t)
        &= 1
    \end{align*}
    We can write this with the asymptotic relation
    \begin{align*}
        \delta \varepsilon \Ei(z_t)
        &\sim (\log t)^{1+\varepsilon}.
    \end{align*}
    By the lower bound shown above, we must have $z_t \ge t \tilde p_t(\delta) = \Omega(\log \log t)$, meaning $z_t \to \infty$.
    By \Cref{lem:Ei-asymptotic}, if $z_t \to \infty$, then $\Ei(z_t) \sim e^{z_t}/z_t$.
    Therefore,
    \begin{align*}
        \delta \varepsilon \Ei(z_t)
        &\sim (\log t)^{1+\varepsilon} \\
        \delta \varepsilon \frac{e^{z_t}}{z_t}
        &\sim (\log t)^{1+\varepsilon} \\
        e^{z_t - \log z_t}
        &\sim \frac{(\log t)^{1+\varepsilon}}{\delta \varepsilon} \\
        z_t - \log z_t
        &\sim \log\paren{\frac{(\log t)^{1+\varepsilon}}{\delta \varepsilon}}
        \tag{\Cref{lem:asmyp-log}} \\
        z_t
        &\sim (1+\varepsilon) \log \log t.
    \end{align*}

    Because $t \tilde p_t(\delta) \le z_t$,
    \begin{align*}
        \lim_{t \to \infty} \frac{\tilde p_t(\delta)}{\frac{\log \log t}{t}}
        \le \lim_{t \to \infty} \frac{z_t}{\log \log t}
        = 1+\varepsilon.
    \end{align*}
    The lower bound we began with yields
    \begin{align*}
        \lim_{t \to \infty} \frac{\tilde p_t(\delta)}{\frac{\log \log t}{t}} \ge 1,
    \end{align*}
    completing the proof.
\end{proof}

\begin{lemma}
    \label{lem:asymp-simplify}
    For a sequence $\{a_t\}_{t=1}^{\infty}$, if $\lim_{t \to \infty} a_t = 0$, then
    \begin{align*}
        1 - e^{a_t} \sim - a_t.
    \end{align*}
\end{lemma}
\begin{proof}
    We must show that
    \begin{equation}
        \lim_{t \to \infty} \frac{1 - e^{a_t}}{-a_t} = 1.
        \label{eq:asymp-goal}
    \end{equation}
    We proceed as follows.
    \begin{align*}
        \lim_{t \to \infty} \frac{1 - e^{a_t}}{-a_t}
        &= \lim_{t \to \infty} \frac{e^{a_t} - 1}{a_t} \\
        &= \lim_{u \to 0} \frac{e^{u} - 1}{u}
        \tag{$\lim_{t \to \infty} a_t = 0$}
        \\
        &= \lim_{u \to 0} \frac{e^{0 + u} - e^0}{u} \\
        &= \frac{d}{du} e^u \bigg|_{u = 0} = 1.
    \end{align*}
\end{proof}

\begin{lemma}
    \label{lem:Ei-asymptotic}
    As $z \to \infty$,
    \begin{align*}
        \Ei(z) \sim \frac{e^z}{z}.
    \end{align*}
\end{lemma}
\begin{proof}
    \begin{align*}
        \lim_{z \to \infty} \frac{\Ei(z)}{\frac{e^z}{z}}
        &= \lim_{z \to \infty} \frac{\frac{d}{dz} \Ei(z)}{\frac{d}{dz} \frac{e^z}{z}} \\
        &= \lim_{z \to \infty} \frac{\frac{e^z}{z}}{\frac{z e^z - e^z}{z^2}} \\
        &= \lim_{z \to \infty} \frac{1}{\frac{z-1}{z}} \\
        &= 1.
    \end{align*}
\end{proof}

\begin{lemma}
    Consider sequences $\{a_t\}_{t=1}^{\infty}, \{b_t\}_{t=1}^{\infty}$ that both go to $\infty$ as $t \to \infty$.
    If $a_t \sim b_t$, then $\log a_t \sim \log b_t$.
    \label{lem:asmyp-log}
\end{lemma}
\begin{proof}
    \begin{align*}
        \lim_{t \to \infty} \frac{\log a_t}{\log b_t}
        &= \lim_{t \to \infty} \frac{\log \paren{b_t \cdot \frac{a_t}{b_t}}}{\log b_t} \\
        &= \lim_{t \to \infty} \frac{\log b_t + \log \paren{\frac{a_t}{b_t}}}{\log b_t} \\
        &= \lim_{t \to \infty} 1 + \frac{\log\paren{\frac{a_t}{b_t}}}{\log b_t} \\
        &= 1 + \frac{0}{\infty} \\
        &= 1.
    \end{align*}
\end{proof}

\end{document}